%% file: main.tex
\documentclass[12pt]{article}
\usepackage{subcaption}
\usepackage[latin1]{inputenc}
\usepackage{amsmath}
\usepackage{amsfonts}
\usepackage{amssymb}
\usepackage{graphics}
\usepackage{multirow}
\usepackage{wrapfig}
\usepackage{epsfig}
\usepackage[margin=1in]{geometry}
\usepackage{color}
\usepackage[table]{xcolor}
\usepackage{booktabs}
\usepackage[export]{adjustbox}
\usepackage{listings}
\usepackage[framemethod=tikz]{mdframed}
\usepackage{siunitx}
\usepackage{hyperref}
\usepackage{xspace}
\usepackage{soul}
\usepackage{algorithm}
\usepackage[noend]{algpseudocode}
\usepackage[export]{adjustbox}
\usepackage{amsthm}
\usepackage{xfrac}
\usepackage[]{graphicx}

\newcommand{\blind}{0}

\newtheorem{lemma}{Lemma}
\newtheorem{corollary}{Corollary}

\makeatletter
\protected\def\ccell#1#{%
  \kern-\fboxsep
  \@ccell{#1}%
}
\def\@ccell#1#2#3{%
  \colorbox#1{#2}{#3}%
  \kern-\fboxsep
}
\makeatother

\newcommand{\gudr}{\ccell{blue!25}}
\begin{document}
\def\spacingset#1{\renewcommand{\baselinestretch}%
{#1}\small\normalsize} \spacingset{1}

\if0\blind
{
  \title{\bf Robust Detection of Covariate-Treatment Interactions in Clinical Trials}
  \author{        
    Baptiste Goujaud,
    Eric W. Tramel,
    Pierre Courtiol,\\
    Mikhail Zaslavskiy\thanks{Corresponding Author (\texttt{mikhail.zaslavskiy@owkin.com})},    
    and Gilles Wainrib \\
    Owkin, Inc.\\
    New York City, NY}
  \maketitle
} \fi

\if1\blind
{
  \bigskip
  \bigskip
  \bigskip
  \begin{center}
    {\LARGE\bf Robust Detection of Covariate-Treatment Interactions in Clinical Trials}
\end{center}
  \medskip
} \fi

\begin{abstract}
Detection of interactions between treatment effects and patient descriptors in clinical trials 
is critical for optimizing the drug development process. The increasing volume of data accumulated 
in clinical trials provides a unique opportunity to discover new biomarkers and further
the goal of personalized medicine, but it also requires innovative
robust biomarker detection methods capable of detecting non-linear, and sometimes
weak, signals.
We propose a set of novel univariate statistical tests, based on the theory of random
walks, which are able to capture non-linear and non-monotonic covariate-treatment interactions. 
We also propose a novel combined test, which leverages the power of all of our
proposed univariate tests into a single general-case tool. We present results for both
synthetic trials as well as real-world clinical trials, where we compare our method with 
state-of-the-art techniques and demonstrate the utility and robustness of our approach.
\end{abstract}

\noindent%
\vfill

\newpage
\spacingset{1.45} 

\section{Introduction}
\label{sec:intro}

Designing new and efficient therapies is a long and ever more costly process, with 
less than ten percent
of new treatments entering Phase I finally being approved by the FDA and commercialized 
\cite{TBA2016,Har2016}. One of the major challenges for the improvement of drug development 
is to better understand how drugs interact with patients, particularly for treatments displaying 
heterogeneous responses. Therefore, conducting a detailed analysis of clinical trial data is 
critical to find subgroups of patients with higher benefit-risk ratio or to understand why a 
drug does not work on some sub-population to improve existing therapeutic strategies.

Moreover, understanding the relationships of patient descriptors which compose the most responsive 
cross-section of the population is of great importance when planning a Phase III trial, for 
salvaging failed trials, or accelerating advances in personalized medicine. This process of 
biomarker identification is critical to detect sub-groups within a given indication, but, 
as shown recently for immunotherapies, can also provide the basis for pan-indication drug 
approval \cite{fda2017}.

Identifying these patient subgroups, and understanding the descriptors, or \emph{covariates}, 
which distinguish them, is the domain of \emph{subgroup analysis}. This field of research has a 
long history in clinical biostatistics \cite{PAE2002, HWB1977, Sim1982}, and requires a very 
careful and rigorous methodological approach \cite{Wang:2007aa}, both for confirmatory or 
exploratory analysis, due to multiplicity, reproducibility, and false detection issues. 
Furthermore, subgroup analysis has garnered even more attention recently from the pharmaceutical 
community with both the advent of cheap and extensive genomic measurements as well as the dawning 
of the era of so-called \emph{big-data}. Now that incredibly detailed patient characterizations 
are available, the problem of subgroup identification has shifted from careful statistical analysis 
of select covariates, to data-mining and machine learning approaches. While the number of features 
is increasing by several order of magnitudes, the number of patients remains typically the same, 
which makes the statistical challenge even more difficult. Therefore, effective clinical trial data
analysis requires highly sensitive tools capable of detecting covariate-response associations from 
noisy and weak treatment response measurements, including discontinuous and non-monotonic 
interactions, without inflating the Type-I error rate. Indeed, false positive detection is a 
major caveat in subgroup analysis \cite{Wang:2007aa}, and controlling Type-I error is a crucial 
requirement for any reliable methodology. These tools can be used as a pre-processing step, 
pruning out insignificant covariates which might obfuscate accurate subgroup identification. 
They can also enable investigators to rank covariates by significance so as to focus laboratory 
studies to a few of the most promising biological pathways.

Many methods have been proposed for the detection of covariate-treatment interaction. 
In particular, we note  modified outcome regression \cite{tian2014simple}, outcome weighted 
learning \cite{ZZR2012}, and change-point detection \cite{koziol1996changepoint} methods as 
state-of-the art techniques which are useful in constructing baseline comparisons, as we do in 
this work. We refer the reader to the recent review of Lipkovich \emph{et al.} 
\cite{Lipkovich:2017aa}, where the authors review various methodologies ranging from classic 
statistical approaches to more sophisticated machine learning methods. On the one hand, 
well-founded and analytically rich statistical techniques provide rigorous estimates of Type-I 
error, but often miss complex, non-linear interactions. On the other hand, machine learning 
approaches allow for the detection of these complex, and sometimes non-monotonic, interactions. 
However, in general they fail to provide a proper estimate of Type-I error for the impact of 
individual covariates; they do not possess the characteristics common to proper statistical tests. 
Thus, we note a lack of methods which take the best of both worlds, offering high sensitivity for 
interactions with complex dependencies, while also providing rich statistical analysis and a 
controlled Type-I error rate.

To address this shortfall, we propose a new series of statistical tests, each of which is designed 
to detect particular structures in the treatment response signal which are often observed in actual 
clinical trial data. The proposed tests are constructed from a cumulative process on the effect 
size obtained by ranking patients according to a given covariate. These processes  characterize the 
complex dependency  structure of the covariate-treatment effect. We introduce several observables 
which capture various facets of these interactions, going beyond a simple process maximum estimate, 
as used in \cite{koziol1996changepoint}. When possible, we derive theoretical estimates to compute 
$p$-values thresholds characterizing Type-I error, and when not possible, we  propose a Monte-Carlo 
sampling procedure. We also present a new \emph{combined} test which leverages the power of all our 
proposed individual tests to provide robust detection of significant covariates, while also 
providing fine-grained control over the Type-I error rate. Our novel  combined test compares 
favorably to existing state-of-the-art procedures, serving as an effective tool for the 
exploration of clinical trial data. We evaluate our approach on both synthetic and real 
clinical trial datasets. For the purpose of this work, we have also created a synthetic benchmark 
that captures many ``corner cases'', i.e. parameter regimes where existing methods reveal their 
limitations. These benchmarks may prove  useful for other researchers to evaluate their methods.

The paper is organized as follows. In Sec. \ref{sec:response_model}, we present the 
treatment-response model and introduce some notations. Next, in Sec. \ref{sec:background}, 
we review the current state-of-the-art in covariate-response correlation analysis and covariate 
selection. Subsequently, in Sec. \ref{sec:centering}, we demonstrate a simple correction based 
on treatment response correlations which addresses many variance issues observed in current 
techniques. Then, in Sec. \ref{sec:cumulative-tests}, we present our novel individual cumulative 
response tests based on null-tests against Brownian motion, each of which is tailored to different 
features in the measured treatment response. Finally, in Sec. \ref{sec:combtest} we present an 
approach to merge these individual cumulative tests into a single combined test which is able to 
retain the performance of the individual tests while being robust to \emph{a priori} unknown 
treatment response curves. To validate our approach, we present two numerical analyses. In our 
first synthetic analysis, reported in Sec. \ref{sec:synth-exp}, we report objective results of 
significant covariate detection performance over a wide range of different treatment response 
models. In our second analysis, reported in Sec. \ref{sec:realexp}, we take three real-world 
clinical trial datasets and report the significant variables discovered by our
method.  Finally, we conclude in Sec. \ref{sec:discussion} with a discussion of the
applicability of our approach and possible avenues for future work.

\section{Drug response model}
\label{sec:response_model}
Let us now describe the mathematical framework of the drug response model. The drug
response model characterizes the observed outcome of a patient as a function of the
patient's covariates and the treatment which was administered to the patient. We will
assume that this observed outcome is binary or  real-valued measurement, 
e.g. cell counts or cholesterol levels.

Let us denote the treatment indicator as $T$. In the common setting of experimental
treatment versus placebo, $T$ is a binary variable with  $T = -1$ corresponding
to the placebo and, $T = 1$ to the experimental treatment. 
Given a vector of covariates $X$ belonging to a particular
patient, we denote the observed outcome under treatment $T = t$ as
as $R^{(t)}(X)$.

A trial dataset consists of $N$ patient records, which are assumed to be drawn 
\emph{i.i.d.} randomly, each of which contains the 
patient's covariate vector, the applied treatment, and the observed outcome, i.e.
$(X_i, T_i, R_i)_{1\leq i \leq N}$ where we define $R_i \triangleq R^{(T_i)}(X_i)$ for conciseness. 

We are interested in the detection of patient covariates which correlate
with the \emph{spread} between the experimental treatment and the placebo, 
or \emph{treatment effectiveness}
\begin{equation}
  E(X) = R^{(1)}(X) - R^{(-1)}(X).
  \label{eq:spread}
\end{equation}

Such covariates can be helpful in the understanding of the treatment action mechanism and
in the selection of patient subgroups where the treatment is the most efficient.

Of course, the spread is not directly observable in practice, as it would require two 
treatments to be carried out on the same patient, so we need special methods capable to
estimate the correlation of interest from indirect measurements.    

In the next section, we revise existing approaches that can be used  for the detection of
covariate-treatment interactions.

\section{Existing approaches}
\label{sec:background}
The problem of identifying covariate-treatment interaction is directly related
to the sub-group identification problem and methods developed for one can be
often adapted for another.  Many different approaches have been proposed to
address these problems, starting from straightforward application of the
standard multivariate regression techniques (with explicit modelling of
the treatment-covariate interactions), to more advanced models such as
modified covariates \cite{tian2014simple}, outcome weighted learning
\cite{ZZR2012} or change point statistics analysis \cite{koziol1996changepoint}.
In \cite{Tian:2011aa}, Tian \& Tibshirani describe an adaptive index model
which can be used for risk stratification or sub-group selection. Other examples
of decision tree based algorithms are model based recursive partitioning
\cite{Zeileis:2008aa}, SIDES method based differential effect search
\cite{Lipkovich:2011aa}, virtual twins method \cite{Foster:2011aa}, subgroup
analysis via recursive partitioning, combined additive and tree based regression
\cite{Dusseldorp:2010aa} and qualitative interaction trees
\cite{Dusseldorp:2014aa}. 

In the following subsection, we describe standard statistical
procedures as well as existing state of the art approaches that can be used to detect
features correlated with the treatment effect.

\subsection{Linear Regression Test}
\label{sec:lineartest}
One na{\"i}ve approach for the detection of significant covariates would be to simply apply a 
linear regression to the observed response variables and study the magnitude of the 
regression coefficients learned for each covariate, using, for instance, the F-test
or $R^2$ values. To construct such a regression, one first defines a linear observational
model for the measured response as a function of the covariate-treatment pair $(X, T)$,
\begin{equation}
  R = \alpha^TX + T\cdot\beta^TX + \epsilon,
\label{eq:linregtest}
\end{equation}
where $\epsilon$ is a centered random variable independent w.r.t. $(X,T)$. 
The vector of regression coefficients 
$\alpha$ describes the \emph{global} impact of patient covariates on the outcome,
irrespective of the applied treatment. This term is often referred to as the \emph{trend} of
the treatment response, and does not contain any information on treatment effectiveness. 
The significant coefficients the vector $\beta$ can be used as indicators of covariates which 
correlate with the treatment effect. In the univariate context, when we analyse
a particular variable $X^j$, (\ref{eq:linregtest}) is simplified to  a simple linear regression with only two
terms:        
\begin{equation}
  R = \alpha_j^TX^j + T\cdot\beta_j^TX^j + \epsilon.
\label{eq:linregtest}
\end{equation}

\subsection{Modified Outcome}
\label{sec:modoutcome}
The problem with the statistical test for the observational model defined in Eq. \eqref{eq:linregtest} 
one must estimate the coefficients $\alpha$ of the trend term jointly with the coefficients of $\beta$. 
Knowledge of $\alpha$ does not aid in the detection of covariate/treatment interaction, and it harms
the estimation of $\beta$, the true variable of interest, by introducing additional variance.

In the modified outcome approach of \cite{tian2014simple}, the authors propose a simple modification to
the observed response prior to regression which removes the effect of the trend term, allowing for the
direct estimation of the expected treatment effectiveness given a set of patient covariates,
$\mathbb{E}\left[ E(x) | X = x\right]$. The modification of the outcome, in the case of two 
treatments, consists of a multiplication of the observed response with the treatment indicator to
create the \emph{modified outcome}, $Y \triangleq R \times T$. In \cite{tian2014simple}, it is shown
that $\beta$ can then be estimated by a regression performed according to the observational model,
\begin{equation}
  Y = \beta^T X + \epsilon.
\label{eq:modoutcome}
\end{equation}

Due to the nature of clinical trials, for each draw of $X$, we have only one observation of $R$ which is drawn from
one of the potential arms of the trial at random (e.g. treatment or placebo). Thus, the modified outcome variable
$Y$ may have a very large variance. Indeed, even if we assume that the per-treatment responses have small variance, the modified outcome $Y = R\times T$ may  then contain multiple distinct tight modes, 
corresponding to each realization of the treatment variable.

These modes have the effect of inflating the variance of $Y$,
\begin{align}
\mathbf{var}[Y|X] &= \mathbb{E}[\mathbf{var}[Y|X,T]~|X] 
 + \mathbf{var}[\mathbb{E}[Y|X,T]~|X].
\label{eq:cond_var}
\end{align}    
Even if $\mathbf{var}[Y|X, T]$ is small for each treatment arm, 
thus causing the expectation of the variance to be small as well, the variance of the expectation can be arbitrarily
large, especially if the treatment in question has a very strong global effect
or if there is a strong shift in outcome distribution. We propose one possible approach to
counteract this variance inflation in Sec. \ref{sec:centering}.

\subsection{Outcome weighted learning}
\label{sec:owl}

The outcome weighted learning (OWL)~\cite{ZZR2012} approach was initially proposed for the 
identification of
patient sub-groups, but similarly to modified covariates it can be easily adapted for the
detection of covariate treatment interactions. The idea of this approach is based on the
construction of a classifier $f(x)$ for the following weighted classification problem         
\begin{equation}
 \underset{f}{\mbox{argmin}}\sum_{i=1}^{N} w_i L(T_i, f(X_i)),\mbox{where} 
 \label{eq_owl}
\end{equation}
$L$ is a classification loss function (hinge loss, for example)  and $w_i$ are positive weights
obtained from observed outcomes $R_i$, for example, $w_i = R_i - \mbox{min}_{i = 1..N} R_i$. It
wouldn't make any sense to try to predict $T$ from $X$ since by construction, $T$ is
generated to be independent of $X$.  However, when we introduce weights, the classifier
tries to separate above all, treatment and placebo patients with high outcome values, and
it might be possible, if there is a pattern in patient to treatment response.
Interestingly, if we adapt this to the univariate case, this approach becomes
equivalent to modified outcome since basically we interested in a correlation
between $T$ and $X^j$ weighted by $R$ which is nothing else than the
correlation between $RT$ and $X^j$. 

\subsection{Discontinuous Treatment Response}
\label{sec:max-cumulative}
Until now, we have discussed the use of linear models for the detection of covariate/treatment 
interaction. However, this is a very simple assumption to make about the treatment response. 
In some cases, a strong non-linear, discontinuous thresholding effect can be the dominant 
feature in the covariate-treatment response curve, as reported by Koziol \& Wu in \cite{koziol1996changepoint}
for the case of erythropoietin treatment (r-HuEPO) for the prevention of post-surgery blood transfusion. 
Here the authors observed sharp cutoffs in treatment effectiveness as measured against 
baseline hemoglobin levels. 

In order to detect responsive patient profiles for r-HuEPO treatment, the authors of \cite{koziol1996changepoint}
proposed to build a stochastic cumulative process test to detect the change-point in measured baseline hemoglobin. 
This univariate test is constructed by building a test around a \emph{cumulative process} description of the 
treatment effectiveness. Specifically, both the placebo and r-HuEPO treatment response curves were sorted according 
to the value of the measured baseline hemoglobin for each patient record. Subsequently, a cumulative sum was taken
for each treatment, and a test was constructed to observe the statistical significance of the difference between
these cumulative sums under a specified threshold of baseline hemoglobin. 

To construct the test itself, \cite{koziol1996changepoint} made the observation that under the 
null-hypothesis ($\mathcal{H}_0$) of no covariate-treatment interaction, the treatment effectiveness should behave as a
random walk when the scale of the covariate is mapped to $[0,1]$, as the observed response would, in this case,
be independent of the covariate and its ordering. In the limit of of $N\rightarrow\infty$, this random
walk converges to a Brownian motion process. 

If we denote this random process under $\mathcal{H}_0$ as $W$, then the detection of 
statistically significant treatment effectiveness amounts to the detection of the measured
cumulative process significantly diverging from $W$. In \cite{koziol1996changepoint}, the 
authors proposed a statistical test that was constructed around the \emph{maximum} value of 
the measured cumulative process as compared to the most probable maximum value according to $W$.
In Sec. \ref{sec:cumulative-tests}, we will demonstrate a version of Koziol \& Wu's 
cumulative test, but within the modified outcome framework, and we will also 
present a series of new tests which are also based on the comparison between cumulative
response and random process null-hypotheses.    

\section{Proposed Approach}
Similar to the approach of~\cite{koziol1996changepoint} which we described in 
Sec.~\ref{sec:max-cumulative}, we focus on univariate detection of significant covariates through
statistical tests based on a null-hypothesis of random walks. Since real-world trial data can
contain many pathological and idiosyncratic response features, we go further than 
\cite{koziol1996changepoint}
by introducing, in Sec. \ref{sec:cumulative-tests},
a set of statistical tests which are designed to detect different features
in the underlying response signal. Since it is impossible to predict \emph{a priori} what 
the best hypothesis of the trial data response should be, we also propose the use of a 
\emph{combined} statistical test. In this way, we are able to utilize our proposed 
statistical tests as \emph{feature detectors}, whose outputs, taken as a whole, create a 
robust description of covariate significance.

Because of the specific structure of the test 
framework, namely repeated tests on individual variables under controlled randomization, we 
are able to do better than simple Bonferoni correction \cite{Dun1961} for combining $p$-values 
obtained over multiple tests. We note that the combined test we 
propose in Sec. \ref{sec:combtest} can be utilized independently of the specific tests we 
construct, and can be seen as a general procedure for obtaining robust predictions of covariate
significance in the setting where one posses many possible statistical tests but requires an
interpretation of the aggregate results.

First, however, we turn our attention to transformations of the 
raw response data. As first shown in the modified outcome approach of \cite{tian2014simple}, such 
transformations can lead to significant improvements in the detection of treatment
response. We will show how one such transformation can both improve the signal-to-noise
ratio for significant covariate detection, and also leads to new significance test for
the cumulative process approach.
    
\subsection{Centered Treatment Response}
\label{sec:centering}         
As discussed earlier in Sec. \ref{sec:modoutcome}, it is possible that the separation of the
treatment response curves can induce high variance in the treatment effectiveness, even when
the per-treatment response variance is small. To counteract this effect, we propose the use of 
a per-treatment centering, removing the empirical average of the response conditioned
on the treatment applied. This approach is similar in spirit to that of efficiency augmentation
\cite{tian2014simple}. In the efficiency augmentation approach,
one attempts to reduce estimator variance by removing the trend  
$\mathbb{E}[R|X]$, while in our approach we remove
$\mathbb{E}[R|T]$. In general, it is possible to combine both approaches,
however, we consider only per-treatment centering since explicit model fitting of $R(X)$ may represent a significant risk of over-fitting on 
trials with limited enrollment, where selecting a trend model \emph{a priori} 
may introduce unnecessary bias. 

For per-treatment centering, 
given the set of observed trial data $(X_i, T_i, R_i)$, we modify the measured
response,
\begin{equation}
  \widetilde{R}_i \triangleq
  R_i - \mathbb{E}[R~|~T = T_i],
\end{equation}
where $\mathbb{E}[R~|~T = T_i]$ is simply the empirical mean of all observed responses for a 
given treatment $T_i$.
Next, for the case of two-arm trials, we apply the same modified outcome approach of 
\cite{tian2014simple} to remove the  trend term by taking the difference of the two treatments,
\begin{equation}
  \widetilde{Y}_i \triangleq
  \widetilde{R}_i \times T_i.
\end{equation}

At first glance, such per-treatment centering would seem to remove the treatment effectiveness
signal. However, we note that we are not interested in the magnitude of the treatment
effectiveness itself, but rather the correlation of the treatment effectiveness with the
covariate of study. This correlation is preserved by the centering, and becomes more 
easily detectable as the variance of the outcome terms is reduced. 

To help motivate our choice of per-treatment centering, we first demonstrate that 
per-treatment centering minimizes the variance of the modified outcome.
%
%
\begin{lemma}
\label{lemma:var-reduce}
Given a randomized trial with each treatment chosen i.i.d. with non-zero probability 
and independently of $X$, 
for the set of all possible modified outcomes of the form $\widehat{Y} = T\cdot(R + f(T))$,
where $f(T)$ is an arbitrary function of the treatment, the function
$$ f(T) = -\mathbb{E}[R | T]$$ 
provides the minimum achievable variance, 
$$\mathbf{var}[\widehat{Y}] = \mathbb{E}[T^2 \cdot \mathbf{var}[R|T]].$$
\end{lemma}
\begin{proof}
See Appendix \ref{apdx:proof-var-reduce}.
\end{proof}

Additionally, we can observe that the original modified outcome of $Y_{\rm mod} = T\cdot R$,
as proposed in \cite{tian2014simple}, will always have a larger variance, except in the 
case that the per-treatment responses are already centered.
%
%
\begin{lemma}
\label{lemma:two-trial-compare}
Given a randomized two-treatment trial, $T \in \left\{\pm 1 \right\}$, whose measured data
has the empirical mean and variance $\mu_T$ and $\sigma_T^2$, respectively, where each 
treatment is chosen with probability $\pi_T$, then  
$\mathbf{var}[\widetilde{Y}] \leq  \mathbf{var}[Y_{\rm mod}]$, since 
$\mathbf{var}[\widetilde{Y}] = \gamma  \cdot \mathbf{var}[Y_{\rm mod}]$ where 
$0 \leq \gamma \leq 1$. 
\end{lemma}
\begin{proof}
See Appendix \ref{apdx:proof-lemma-two-trial-compare}.
\end{proof}

From Lemma \ref{lemma:two-trial-compare}, we see that, in the case of two treatment trials, the 
variance reduction provided by per-treatment response centering becomes more pronounced as 
the separation in expected response between the trials increases and the intra-treatment variance
decreases. This shows the corrective effect of centering in correcting for the problem of 
multi-modal response distributions, as discussed earlier in Sec. \ref{sec:modoutcome}.

Now that we have established that per-treatment centering is an effective form of variance reduction
for the modified outcome, the question of treatment-response detection remains. We now investigate 
the effect of centering on the correlation between the modified response and treatment.
\begin{lemma}
\label{lemma:centering}
The centering $\widetilde{R} = R - \mathbb{E}[R | T]$ only alters the correlation between
treatment and outcome for a specified covariate by a constant amount which does not depend on
$X$,
$$\mathbf{cov}[\widetilde{R}, T | X] = \mathbf{cov}[R, T | X] + \mathcal{C}_{R,T}.$$

\end{lemma}
\begin{proof}
See Appendix \ref{apdx:proof-lemma-centering}.
\end{proof}
Since covariance between outcome and treatment is modified only by a constant which does
not depend on the covariate $X$, we see that the correlation conditioned on covariates
is not lost by introducing a per-treatment centering, only translated. 

Additionally, we can see that the centered modified outcome removes the effect of
global effects on the covariate-conditioned signal. 
\begin{lemma}
\label{lemma:expected-covariance}
For $\mathbb{E}[T] = 0$,
centering removes global dependence between the measured outcome and treatment,
$\mathbf{cov}[\widetilde{R},T] = 0.$
\end{lemma}
\begin{proof}
    See Appendix \ref{apdx:proof-lemma-expected-covariance}.
\end{proof}
\begin{corollary}
\label{corollary:partial-corr}
Controlling for the ordering of covariates, the partial correlation between the
centered treatment response and the treatment indicator is zero, 
$\rho_{R,T \cdot X} \geq \rho_{\widetilde{R},T\cdot X}  = 0$.
\end{corollary}
From Lemma~\ref{lemma:expected-covariance}, we make the observation that since the global
dependency is removed, then only conditioning on $X$ introduces a dependency between 
the treatment and centered response. Without centering, 
$\mathbf{cov}[R, T] = \mathbb{E}[T\cdot \mathbb{E}[R|T]]$.
This leads to Corollary~\ref{corollary:partial-corr}, which demonstrates that when using centered
treatment response, all of the dependency between $R$ and $T$ is entirely mediated by the
covariate $X$. In other words, for centered treatment response, we see that there remains no
explanation of the covariance other than the effect of the covariate under investigation.
This property provides a greater sensitivity when constructing significance tests, as we can be
assured that the variability being observed in $\widetilde{Y}$ is only due to $X$.

\subsection{Cumulative Tests}
\label{sec:cumulative-tests}
As discussed in Sec. \ref{sec:max-cumulative}, nonlinearities in the treatment response signal 
when conditioned on patient covariates can lead to poor detection of covariate-treatment correlations
when using linear regression fits. Instead, methods such as \cite{koziol1996changepoint} propose
the use of non-linear tests based on a theory of random walks. We now discuss a set of novel 
tests for the detection of significant levels of covariate-treatment interaction according to
$p$-value. 

We first introduce the core transformation utilized by all of our proposed significance tests. 
Given some modified outcomes $Y$, centered or not, and a covariate of interest $X^j$, a sorting permutation
$\mathbf{s} = [s_1, \dots, s_N]$ is constructed such that 
$\bar{X}^j = [X^j_{s_1}, X^j_{s_2}, \dots, X^j_{s_N}]$ is the vector of 
patient covariate values sorted in ascending order. Here, we assume that the covariate in question
has some interpretation that allows for a meaningful sense of ordering. However, it is still possible
to use this procedure in the case of purely categorical covariates, as long as the arbitrarily
chosen ordering remains consistent when repeated on the same covariate.

Subsequently, a cumulative response vector is 
then calculated as the cumulative sum,
\begin{equation}
  C^j_i = \sum_{k=1}^i\quad Y_{s_k}.
\end{equation}
Rather than constructing tests on the modified outcomes themselves, we evaluate
the statistics of the cumulative response process $C^j$. By using the known statistics of 
random walks, we may construct a strong set of priors on the behavior of $C^j$ 
under the null-hypothesis of no covariate-treatment interaction,
\begin{equation}
\mathcal{H}_0: Y \perp X^j.
\end{equation}
Detecting a significant interaction amounts to the detection of $C^j$ diverging significantly
from the bulk where it is most explained by a random process.

\paragraph{Maximum Value Test.}
We first turn our attention to small adaptation of the original cumulative maximum
value test of \cite{koziol1996changepoint}, generalizing their change-point detection test
to the case of interaction detection. We propose the use of the maximum value test as a 
baseline comparison when used in conjunction with the modified outcome of \cite{tian2014simple}.

Specifically, we define the maximum normalized absolute value of the cumulative response as
\begin{align}
  M = \max_{i\in[1,N]}\quad \left\lvert \frac{1}{\sqrt{N \sigma_N^2}}\cdot C_i^j \right\rvert,
\end{align}
where $\sigma_N^2$ is the sample variance taken of the entries of the realized random walk
$C^j$,
\begin{equation}
    \sigma_N^2 = \frac{1}{N-1}\sum_{i=1}^N \left(C^j_i - \sum_{k=1}^N C_k^j\right)^2.
\end{equation}
Next, as mentioned
earlier in Sec. \ref{sec:max-cumulative}, we observe that as $N \rightarrow \infty$, 
according to Donsker's theorem \cite{Don1951}, 
the cumulative process under $\mathcal{H}_0$ converges to a Weiner process. 
The distribution of the extreme values of a Wiener process is well known in the literature 
\cite{koziol1996changepoint,AL2013,BS2002}, 
\begin{align}
\mathbb{P}&\left[\left\{\max_{0\leq t \leq 1} \left\lvert W_t\right\rvert\right\} > \alpha\right] = 
4\sum_{i = 1}^\infty (-1)^{i+1}\Phi(-(2i-1)\alpha), 
\label{eq:maxcum}
\end{align}
for any $\alpha > 0$ where $\Phi$ is the standard normal cumulative distribution function (CDF). The 
final significance test is constructed from \eqref{eq:maxcum} by measuring the probability of the
observed maximum value and comparing it to the probability that the null-hypothesis of the random
walk could have produced such a result. In this case we are assuming that $N$ is large enough
such that the approximation via Weiner process is accurate.

\paragraph{Brownian Bridge Test.}
In the case of the maximum value test, we define this baseline using the standard modified outcome.
However, we would like to make use of the per-treatment response centering we propose
in Sec.~\ref{sec:centering} in order to aid in the detection of covariate-treatment 
correlation. However, by using this per-treatment centering, the cumulative process
is no longer well described by Brownian motion. Once the mean is removed, the cumulative 
process is pinned to 0 not just at $t=0$, but also at the end of the process $t=1$. An 
example of such a process is shown in Fig.~\ref{fig:brownian-excursion}.

Specifically, if we construct the cumulative response using the per-treatment 
centered response,
\begin{equation}
  \widetilde{C}^j_i = \sum_{k=1}^{i}\quad\widetilde{Y}_{s_k}
\end{equation}
then, by application of Donsker's theorem \cite{Don1951},
$\widetilde{C}^j$ converges to a Brownian bridge process, $B_t$, as $N\rightarrow\infty$ 
instead of Wiener process. We can observe this convergence by rescaling the discrete 
range of patient indices $[1, N] \mapsto [0,1]$ via the product $tN$, to find
\begin{align*}    
&\frac{1}{\sqrt{N\sigma^2}}
  \sum_{k=1}^{\lfloor tN \rfloor} \left(\widetilde{Y}_{s_k} +\left(tN - \lfloor tN\rfloor\right) \widetilde{Y}_{s_{\lfloor
  tN\rfloor+1}}\right) 
\underset{N\rightarrow \infty}{\Longrightarrow}
  (B_t)_{0\leq t\leq 1}.
\end{align*}
The same transformation can be applied to the maximum absolute value of the process, as well.

Now, according to Slutsky's theorem, we can replace $\sigma^2$ with its empirical estimate,
\begin{align}
\underset{i \in [1,N]}{\max} \left\lvert
  \frac{1}{\sqrt{N\sigma_N^2}}\cdot \widetilde{C}_i^j \right\rvert
\underset{N\rightarrow \infty}{\Longrightarrow} \underset{0\leq t \leq 1}{\max}
   \left\lvert B_t\right\rvert,
\end{align}
which shows that the scaled extremal value of the centered process converges in distribution to the
extremal of the Brownian Bridge. Thus, we are able to construct the statistical significance test
against $\mathcal{H}_0$ according to the distribution of extreme values of the Brownian Bridge
process.

And the statistical test rejection zone  can be computed according to       
\begin{align}
&\mathbb{P}\left[\left\{\underset{0\leq t \leq 1}{\max}\lvert B_t \rvert\right\} > \alpha\right]
  = 
  2\sum_{i=1}^{\infty} (-1)^{n-1} e^{-2 i^2 \alpha^2}.
  \label{eq:bridge-reject}
\end{align}
While this extremal value test is a robust statistic for significant covariate detection, there are many examples of
significant treatment-covariate interaction patterns which may be missed by constructing 
$\mathcal{H}_0$ from a Brownian Bridge process. In the next
section, we discuss some of these shortcomings and demonstrate an extremal test on Brownian Excursion
processes that succeeds where \eqref{eq:bridge-reject} fails.

\paragraph{Brownian Excursion Test.}    
The maximum absolute value is a good statistic for monotonic signals, where the
treatment effect monotonically increases or decreases with the value of the
covariate under study.
However, when the covariate does not display monotonicity, then maximum value tests may 
produce false negatives with respect to covariate significance.
An example of such a case might be a covariate for which the extreme values, low or high,
are correlated with low treatment response, while mid-range values are correlated with high 
treatment response. The max value cumulative process approach may fail in this case as a particular 
realization in this case may first decrease, rise, and then decrease again without ever reaching
a critical extremal value which would indicate a statistically significant deviation from 
$\mathcal{H}_0$. However, such a process might indeed be judged significant if the origin of the
process were shifted to the beginning of the domain correlated with positive treatment effect. Thus,
a thoroughly general, if cumbersome, approach might be to utilize a max value cumulative test 
evaluated from every possible start position.

Specifically, given a centered cumulative process for a specific covariate $\widetilde{C}$,
consider the family of circle-shifted processes constructed from it,
\begin{equation}
    \mathcal{C} = \{(\widetilde{C})_l~|~\forall l = 0, 1, \dots, N \},
\end{equation}
where 
\begin{equation}
(\widetilde{C}_i)_l = 
    \sum_{k=1}^i \widetilde{Y}_{s_{(k+l)~{\rm mod}~N}}~~.
\end{equation}

Finding the extremal value
for all possible shifts is then the solution to the double maximization, 
\begin{equation} 
    \label{eq:shift-maximization}
    \max_{\widetilde{C}\in\mathcal{C}} \max_{i \in [1, N]}~~\left|\widetilde{C}_i\right|.
\end{equation}
The solution to this maximization over shifts and time can be be equally accomplished by 
generalizing the Brownian bridge process to a ring domain and subsequently defining the 
``start'' of this process to occur at its minimum. Subsequently, finding the extremal of this 
resultant process would be equivalent to \eqref{eq:shift-maximization}. An example of this 
construction is given in Fig.~\ref{fig:brownian-excursion}.

\begin{figure}[t]
    \centering    
    \includegraphics[width=0.49\textwidth]{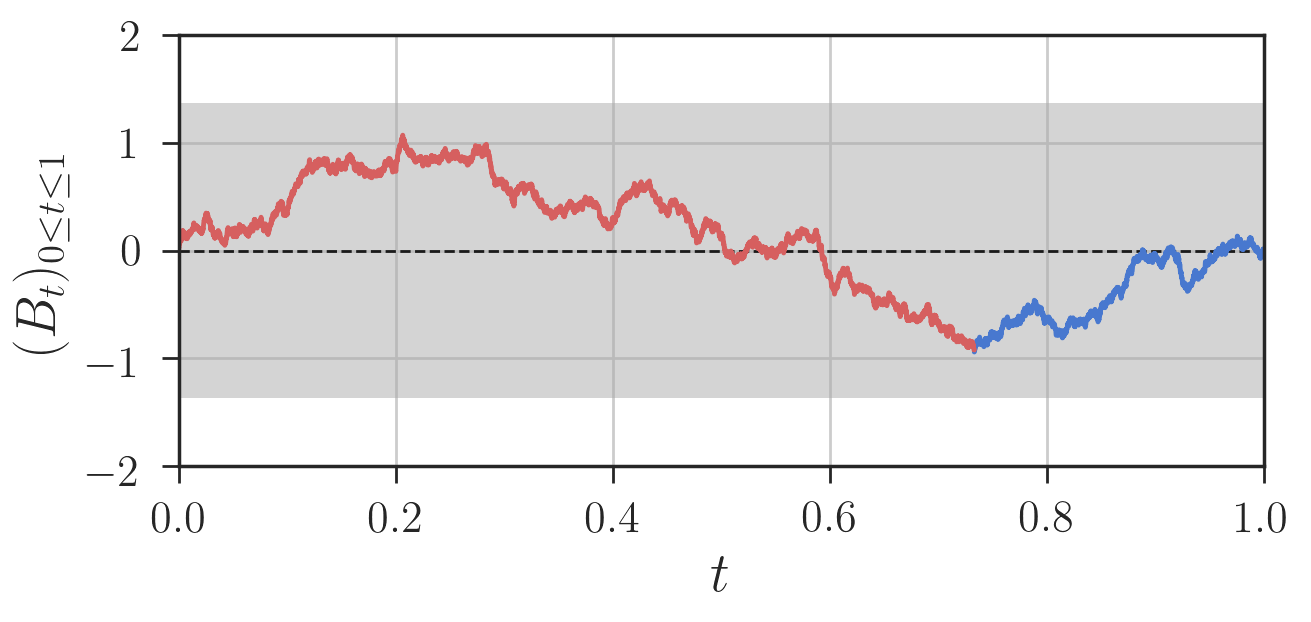}
    \includegraphics[width=0.49\textwidth]{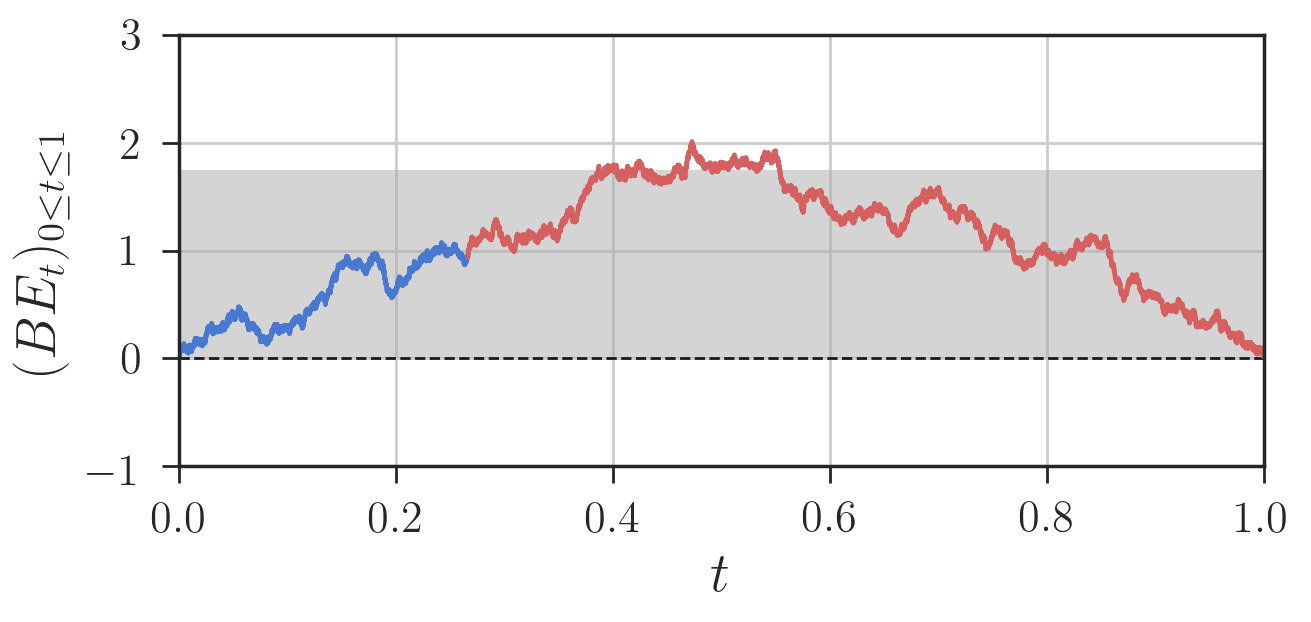}
    \caption{%
       Example realizations of both the Brownian Bridge process (\emph{left}) as well as the Brownian
       Excursion process, here constructed via circular shift of existing
       Brownian Bridge realization (\emph{right}). Coloring indicates the re-arrangement used to 
       construct the Brownian Excursion. For both processes, the 95\% confidence region
       for the process is given in grey. In this case, a Brownian Bridge process well within
       the confidence interval ($\approx [-1.365, 1.365]$) constructs a Brownian Excursion 
       process whose maximum leaves the interval ($\approx[0, 1.745]$).
       \label{fig:brownian-excursion}}
\end{figure}

The procedure we have just described is nothing more than finding the extremal value of a
Brownian Excursion,
\begin{align}
&\underset{0\leq t \leq 1}{\max} BE_t =
\underset{0\leq t \leq 1}{\max}  B_t - \underset{0\leq t \leq 1}{\min}  B_t.
\end{align}
The rejection zone can be computed according to 
\begin{equation}
\mathbb{P}\left[\left\{\underset{0\leq t \leq 1}{\max} BE_t \right\} > \alpha \right] 
  = 2\sum_{i = 1}^{\infty}(4 i^2 \alpha^2 - 1) e^{-2 i^2 \alpha^2}.
\end{equation}

\paragraph{Process Normalization.}
The overall maximum of a Brownian motion, or a Brownian bridge, are natural statistics to 
consider, however
they treat the entire interval $[0,1]$ to be homogeneous. However, it is much more likely to
observe the maximum value at the right extreme for Brownian motion, or the center for the
Brownian bridge. Therefore, even if  there is  a signal in the zone of low variance, we
compare it to the maximum over the entire interval and therefore it might easily be
overlooked. 

This problem can be solved by considering a rejection zone taking into account the
changing variance of the underlying process. In the case of a Brownian motion, such
normalization means that the rejection zone becomes a square root hull
$t\alpha$, and in the
case of a Brownian bridge, the hull shape is defined by $\sqrt{t(1-t)}$. 

\paragraph{Area Tests.}
\begin{wrapfigure}[23]{hr}{0.4\textwidth}
    \vspace{-5ex}
    \scalebox{0.8}{
    \begin{minipage}{0.47\textwidth}
    \begin{algorithm}[H]
    \caption{Single Cumulative Test}
    \begin{algorithmic}
    \State \textbf{Input:} \\
      $X$: Set of Patient Covariates\\
      $T$: Treatment Indicators \\
      $R$: Measured Outcome (endpoint) \\
      $M$: Number of Monte-Carlo Simulations
    
    \\\hrulefill
    \vspace{-1ex}\State \textbf{Preprocessing}
    \\\vspace{-2ex}\hrulefill
    \For{$i \in 1, \dots, N$}
    \State $R_i \leftarrow R_i - \mathbb{E}[R|T=T_i]$
    \State $T_i \leftarrow T_i-\mathbb{E}[T]$ 
    \State $Y_i \leftarrow R_i \cdot T_i$
    \EndFor
    
    \\\hrulefill
    \vspace{-1ex}\State \textbf{$\mathcal{H}_0$ Statistics via MC}
    \\\vspace{-2ex}\hrulefill
    \For{$k \in 1, \dots , M$}
    
    \State $\mathbf{q}$ $\leftarrow$ $\text{RandomPermute([1, 2, \dots, N])}$
    
    \State $C \leftarrow (\sum_{i = 1}^n Y_{q_i})_{1 \leq n \leq N}$
    
    \State $S[k]$ $\leftarrow$ Test statistics on $C$   
    
    \EndFor
    
    \\\hrulefill
    \vspace{-1ex}\State \textbf{Test on True Data}
    \\\vspace{-2ex}\hrulefill
    \State $\mathbf{s} \leftarrow {\rm SortPermutation}(X)$
    
    \State $C \leftarrow (\sum_{i = 1}^n Y_{s_i})_{1 \leq n \leq N}$
    
    \State $V$ $\leftarrow$ Test statistics on $C$
    
    \State \textbf{Output:} \\
      $p\text{-value} \leftarrow
      \frac{1}{M}\sum_{k=1}^M{\mathbf{1}_{S[k]>V}}$
    
    \end{algorithmic}
    \label{alg:single-test}
    \end{algorithm}
    \end{minipage}}
    \end{wrapfigure}
Finally, we propose two more tests based on the total area under the curve (AUC) statistic, the
idea behind the use of the entire area is that even if the maximum value does not reach
the critical threshold, the fact that there are multiple points where the process
\emph{approaches} the limit may be indicative of a presence of a signal.

The test statistic in this case is computed as 
\begin{align}
  A = \sum_{i\in[1,N]}\quad \left\lvert \frac{1}{\sqrt{N \sigma_N^2}}\cdot C_i^j \right\rvert.
\end{align}
We also consider a modification of this test where we sum squared values of the
cumulative process,
\begin{align}
  SA = \sum_{i\in[1,N]}\quad \left ( \frac{1}{\sqrt{N \sigma_N^2}}\cdot C_i^j
  \right )^2.
\end{align}

\paragraph{Implementation.}
Depending on a prior knowledge about the type of the signal, other tests can be
introduced as well.  Since it is not always possible to have a closed form for the
statistic distribution, we use a general numerical framework based on Monte-Carlo
(MC) simulations to compute statistic critical values \cite{Shao:2012aa, Higgins:2003aa}. 
We describe the process for this evaluation in Alg. \ref{alg:single-test}.

\subsection{Combined Test}
\label{sec:combtest}
            
As mentioned in the previous section, one may prefer different tests depending
on the type of the signal one expects to observe in the data. However, it is not
always possible, or desirable, to choose a particular test \emph{a priori}.
Instead, one might run all tests in parallel to see if some signal is
detected by at least a single test. This procedure is a classical
multiple test, and thus one needs to use a correction procedure for calculation of 
the final $p$-value in order to ensure that the Type-I error is not inflated.  
Standard correction procedures, such as Bonferoni correction \cite{Dun1961}, 
may be too severe, obfuscating the detection of significant covariates in practice. 
This is especially true in our case due to high correlation between 
our proposed significance tests, thus
leading to a significant loss in combined test power.  However, as in
\cite{Gates:1991aa}, it is possible to adapt our numeric procedure to run
multiple tests simultaneously without any losses in Type-I or Type-II errors.

\begin{wraptable}{hr}{0.4\textwidth}
    \begin{adjustbox}{max width=0.4\textwidth}
  \begin{tabular}{lp{5.0cm}}
    \toprule
    {\bf Test Name} &  {\bf Description}\\
    \midrule
    Baselines & \\
    \cmidrule(lr){1-2}
    \textit{MoLin}   & Modified outcome\\ & linear regression test\\
    \textit{Max}     & Max of Brownian motion\\ 
    \cmidrule(lr){1-2}
    Proposed & \\
    \cmidrule(lr){1-2}
    \textit{MaxB}    & Max of Brownian bridge\\
    $\text{\it MaxB}_N$  & Max of Brownian bridge, \\ & normalized via $\sqrt{t(1-t)}$\\
    \textit{MaxBE}   &  Max of Brownian excursion\\
    $\text{\it MaxBE}_N$ &  Max of Brownian excursion\\ & normalized via $\sqrt{t(1-t)}$\\
    \textit{AreaB}     &  Area under Brownian bridge\\
    \textit{SAreaB}   &  Squared area under Brownian bridge\\
    \bottomrule
    \end{tabular}
\end{adjustbox}
    \caption{List of statistical tests for the identification of covariate-treatment
    interactions evaluated in this work.
    \label{tab:test-descriptions}}
\end{wraptable}

Specifically, when calculating the combined significance test, 
for each covariate permutation of the centered treatment response, 
we compute the list of metrics defined by each statistical test included in the combined test, 
along with their corresponding $p$-values. Subsequently, we compare the vector of observed 
statistics with samples generated from random permutations of the centered treatment response. 
Then, of all the $p$-values for the individual tests in the combined test,
we use the \emph{minimum} $p$-value as the final aggregated statistic. From this minimum, 
we define the order on the set of vectors with $p$-values and also compute the $p$-value of 
the combined test. We present the pseudo-code for this combined test in Alg.~\ref{alg:combined-test}.

In our experiments, we combine \textit{MaxB}, \textit{MaxB}$_N$, \textit{MaxBE},
\textit{SAreaB}, and \textit{SAreaB} for our final combined significance test, thus
accounting for many possible signal shapes. Selecting this subset, rather than applying \emph{all}
possible tests, reduces the computational burden of running the combined test.
In practice, it is also possible to build a
combined test tailored to \emph{a priori} knowledge about a given dataset, as long as such a 
construction is fixed prior to any analysis of the dataset under investigation.

\needspace{4\baselineskip}
\subsection{Multi-dose trials}
\label{sec:corrcorr}
\begin{wrapfigure}[23]{hr}{0.4\textwidth}
    \vspace{-10ex}
    \scalebox{0.8}{
    \begin{minipage}{0.5\textwidth}
    \begin{algorithm}[H]
\caption{Combined Test}
\begin{algorithmic}

\State \textbf{Input:} \\
  $X$: Set of Patient Covariates\\
  $T$: Treatment Indicators \\
  $R$: Measured Outcome (endpoint) \\
  $M$: Number of Monte-Carlo Simulations\\
  $[F_1, \dots, F_L]$: List of $L$ Tests to Combine

\\\hrulefill
\vspace{-1ex}\State \textbf{Preprocessing}
\\\vspace{-2ex}\hrulefill
\For{$i \in 1, \dots, N$}
\State $R_i \leftarrow R_i - \mathbb{E}[R|T=T_i]$
\State $T_i \leftarrow T_i-\mathbb{E}[T]$ 
\State $Y_i \leftarrow R_i \cdot T_i$
\EndFor

\\\hrulefill
\vspace{-1ex}\State \textbf{$\mathcal{H}_0$ Statistics via MC}
\\\vspace{-2ex}\hrulefill

\For{$m \in 1, \dots , M$}

\State $\mathbf{q}$ $\leftarrow$ $\text{RandomPermute([1, 2, \dots, N])}$

\State $C \leftarrow (\sum_{i = 1}^n Y_{q_i})_{1 \leq n \leq N}$

  \For{$l \in 1, \dots, L$}
    \State $S[l, m]$ $\leftarrow$ $F_l(C)$  
  \EndFor

\EndFor

\\\hrulefill
\vspace{-1ex}\State \textbf{Distribution of Minimum $p$-value}
\\\vspace{-2ex}\hrulefill
\For{$k \in 1, \dots , M$}
  \State $P[k]$ $\leftarrow$ $\underset{l}{\min}\quad \frac{1}{M}\sum_{m=1}^M{\mathbf{1}_{S[l, m]>S[l, k]}}$    

\EndFor

\\\hrulefill
\vspace{-1ex}\State \textbf{Test on True Data}
\\\vspace{-2ex}\hrulefill
\State $\mathbf{s} \leftarrow {\rm SortPermutation}(X)$

\State $C \leftarrow (\sum_{i = 1}^n Y_{s_i})_{\mathbf{1} \leq n \leq N}$

  \State $V \leftarrow \underset{l}{\min}\quad \frac{1}{M}\sum_{m=1}^M{\mathbf{1}_{S[l,m]>F_l(C)}}$

\State \textbf{Output:} \\ 
Combined $p\text{-value} \leftarrow
  \frac{1}{M}\sum_{k=1}^M{\mathbf{1}_{P[k]>V}}$

\end{algorithmic}
  \label{alg:combined-test}
\end{algorithm}
\end{minipage}}
\end{wrapfigure}

Up to this point, we have considered the case of a clinical trial with only two treatment groups 
(e.g. placebo vs. treatment). 
However, in many cases, such as Phase II trials, multiple
treatment doses are tested in parallel. In this section, we demonstrate how our approach may be
generalized to the case of multiple doses trials. The generalization is quite straightforward:
instead of $T\in\{\pm 1 \}$ being a binary variable, we let instead $T\in\mathcal{T}\subset\mathbb{R}$ be a 
real-valued encoding of different doses. The set $\mathcal{T}$ may be a discrete scale encoding only dose order, 
or perhaps a more precise representation of the dosing amount, e.g. log value of the amount of administered treatment.

As in the binary case, the only important condition we need to ensure is that
$\mathbb{E}[T] = 0$. As long as this is true, then we can directly apply the modified outcome transformation 
and use the framework of cumulative processes without any additional modifications. In essence,
in the binary case we focus on the  conditional correlation between $R$ and a centered binary
variable $T$ conditioned on $X$, and in the multi-dose case we do the same thing only with
$T$ being simply centered.      

\section{Synthetic Experiments}
\label{sec:synth-exp}  
In order to evaluate the efficiency of the proposed approach in an objective manner with known
ground-truth, we generate a synthetic dataset of 
clinical trials. Although these models do not reflect the full complexity of such data in
practice, 
we consider synthetic models for two main reasons: 
(i) to validate the ability of the evaluated statistical 
tests to accurately recover covariates which are known to be significantly correlated with the
treatment, and
(ii) to investigate particular cases of linear and non-linear interactions between 
covariates and treatment effects, demonstrating the strengths and weaknesses of each test. 
We list 
all of evaluated statistical tests in Table~\ref{tab:test-descriptions}. 
Additionally, we demonstrate
how to construct synthetic datasets for which each of the proposed test will fail, giving 
insight into the failure cases for each of the proposed statistical tests.
We will also show that the combined test proposed in Sec. \ref{sec:combtest} 
provides near best-case accuracy for most of the tested synthetic models.

\subsection{Synthetic models}      
For our experiments, we construct a wide array of synthetic treatment response
curves whose purpose is to simulate the complex behavior observed in real-world
clinical trials. Below, we describe each of the different synthetic models.

We consider three types of synthetic models, in order of increasing complexity,
({\it i}) a linear model, ({\it ii}) a set of piecewise-constant models, and ({\it iii}) 
a fully non-linear model.

For all these models, we generate covariates as i.i.d. uniform random variables over $[0,1]$. 
Next, given the full set of patient covariates $X$, 
we write the response model in the general form,
\begin{align}
    R[T | X] &= W_{\rm trend}(X) 
    + T\cdot W_{\rm interact}(X) + \epsilon,
\end{align}
where $\epsilon \sim \mathcal{N}(0,\Delta)$ is a noise term with variance $\Delta$, and 
$W_{\rm trend}(\cdot)$ and $W_{\rm interact}(\cdot)$ are functions operating across the 
set of covariates. 

\begin{figure}[t!]
    \centering
    \begin{subfigure}[t]{0.45\textwidth}
        \centering
        \includegraphics[width=\textwidth]{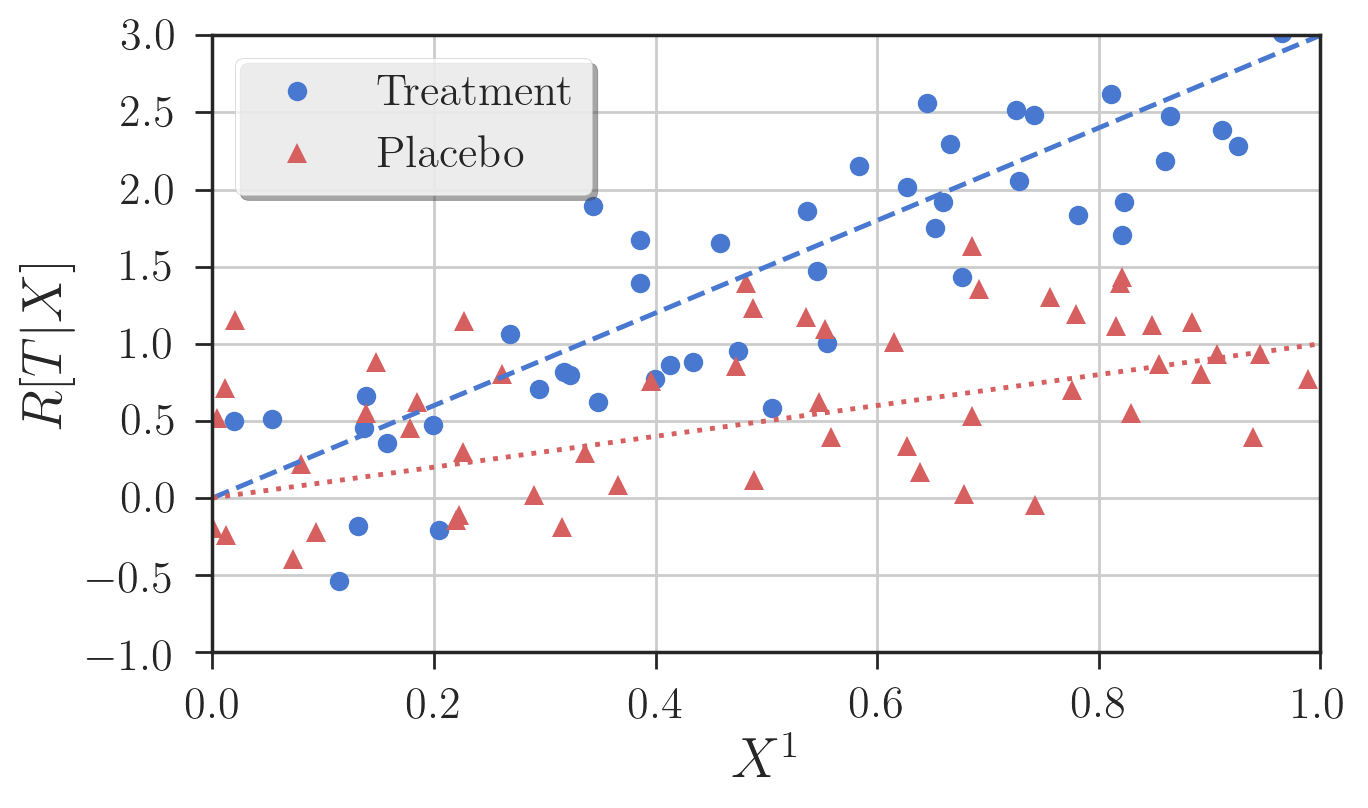}
        \caption{\textbf{L} for $\Delta = 0.25$, $W_1 = 2$, $W_2 = 1$}
    \end{subfigure}
    ~
    \begin{subfigure}[t]{0.45\textwidth}
        \centering
        \includegraphics[width=\textwidth]{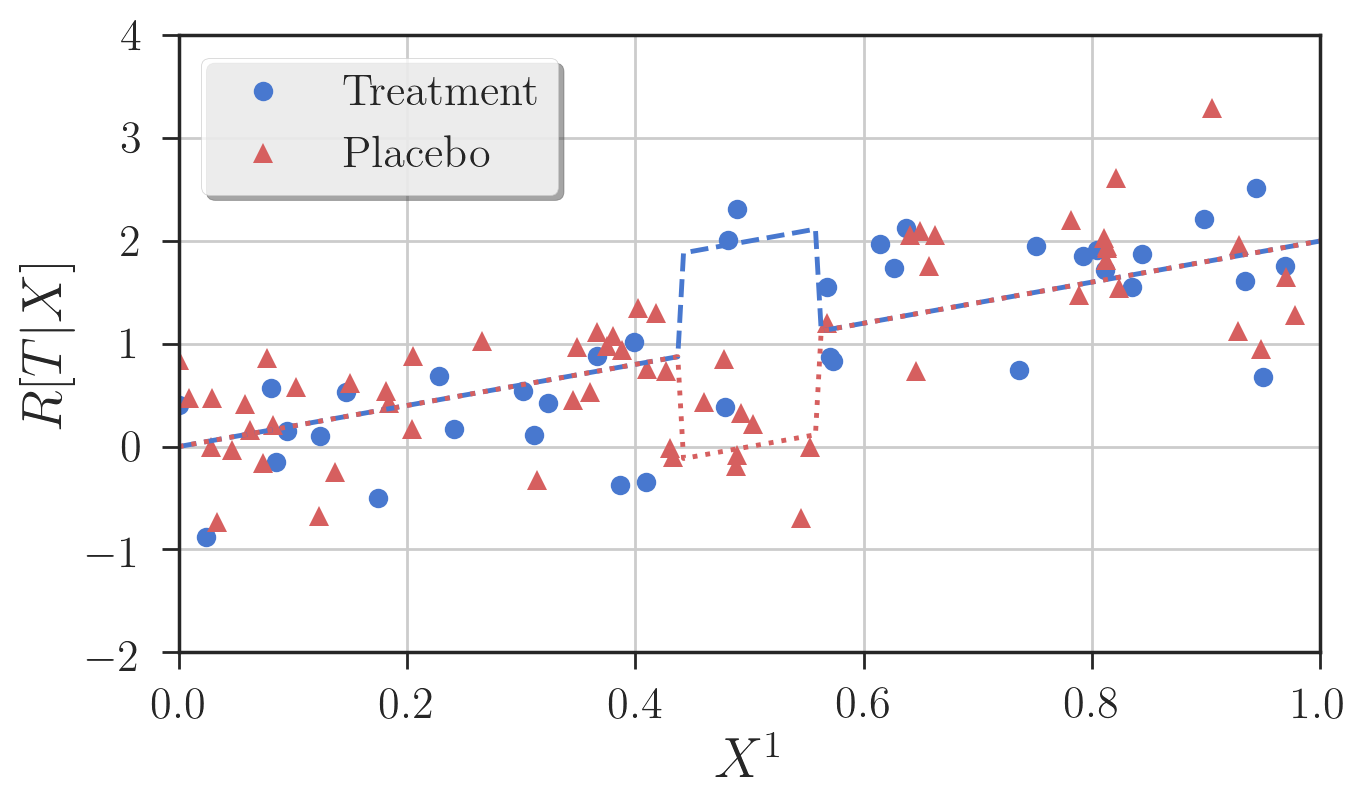}
        \caption{\textbf{PC-Int}$_{2}$ for $\Delta = 0.25$, $W_1 = 2$, $W_2 = 1$}
    \end{subfigure}
    \\
    \begin{subfigure}[t]{0.45\textwidth}
        \centering
        \includegraphics[width=\textwidth]{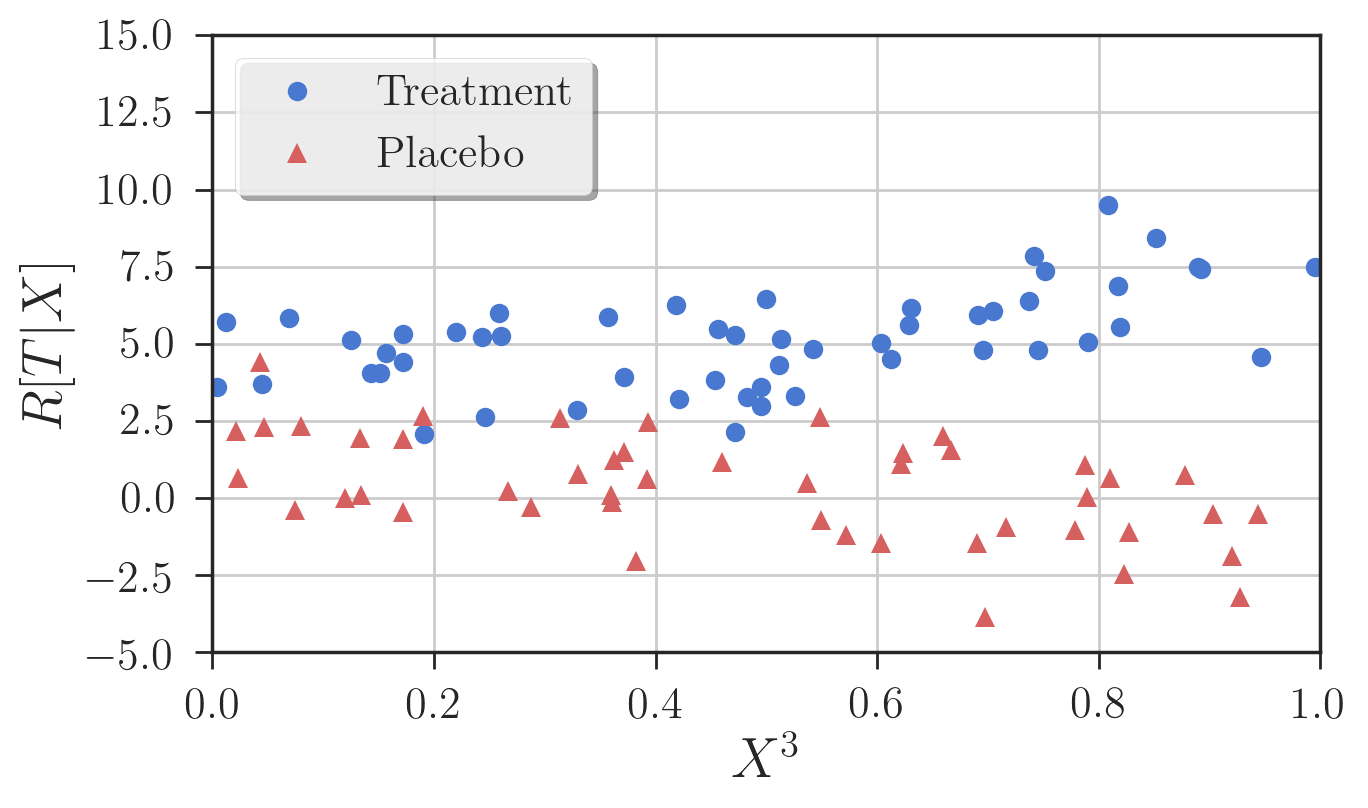}
        \caption{\textbf{NL} for $\Delta = 0.25$}
    \end{subfigure}
    \caption{\label{fig:example-linear-model}%
    Example $R[T|X]$ curves of different synthetic models for $N=100$.}
\end{figure}

\paragraph{Linear Model.}
The linear model (\textbf{L}) corresponds to the case where $W_{\rm trend}$ and 
$W_{\rm interaction}$ are linear functions of 
the of the first covariate, in this case
\begin{align}
    W_{\rm trend}(X) &= W_1 X^1, \\
    W_{\rm interact}(X) &= W_2X^1.    
\end{align}
A successful test for \textbf{L} would report $X^1$ as
a statistically significant covariate in terms of treatment interaction, and all other covariates
would be rejected.

\paragraph{Piecewise-constant Models.}
These proposed piecewise-constant ($\textbf{PC}$) synthetic models contain 
either one or two discontinuous jumps, denoted as \textit{thresholding} and \textit{interval} 
effects, respectively. Such discontinuities 
are not captured in linear models, however such effects are ubiquitous 
throughout biology and medicine.

We introduce four versions of $\textbf{PC}$ models to capture these different configurations. 
All choices of $W_{\rm interaction}(X)$ are simple indicator functions with different supports to
capture both thresholding and interval discontinuities. We present two thresholding models 
containing a single jump discontinuity, \textbf{PC-Th}$_1$ and \textbf{PC-Th}$_2$, as well as
two interval models containing two jump discontinuities, \textbf{PC-Int}$_1$ and
\textbf{PC-Int}$_2$, 
\begin{align}
    W_{\rm interact}(X) &= W_2 \mathbf{1}_{[\sfrac{1}{2}, 1]}(X^1), \tag{{\bf PC-Th}$_1$}\\
    W_{\rm interact}(X) &= W_2 \mathbf{1}_{[0, \sfrac{1}{8}]}(X^1), \tag{{\bf PC-Th}$_2$}\\
    W_{\rm interact}(X) &= W_2 \mathbf{1}_{[\sfrac{1}{4}, \sfrac{3}{4}]}(X^1), \tag{{\bf PC-Int}$_1$}\\
    W_{\rm interact}(X) &= W_2 \mathbf{1}_{[\sfrac{7}{16}, \sfrac{9}{16}]}(X^1). \tag{{\bf PC-Int}$_2$}
\end{align}
Finally, for each of the four different {\bf PC} models, we use a linear baseline trend of 
$W_{\rm trend}(X) = W_1 X^1$.

\paragraph{Non-linear model.} 
For the non-linear model (\textbf{NL}), we utilize the fourth synthetic
scenario reported in Sec. 4 of \cite{ZZR2012}. Specifically, the covariate-treatment
signal is generated according to
\begin{align}
    W_{\rm trend}(X) &= 1+ 2 X^1 + X^2 +0.5 X^3,\\
    W_{\rm interact}(X) &= 1
    - \left(X^1\right)^3  
    + \exp\{\left(X^3\right)^2 + X^5\}
    + 0.6X^6
    - \left(X^7 + X^8\right)^2.
\end{align}
We note that covariate $X^4$ is already a decoy covariate in this definition. Additionally, 
covariates $X^7$ and $X^8$ are symmetric.
        
\paragraph{Experimental Parameters.}
\label{sec:exp-param}
For our experiments, we evaluated the proposed statistical tests on the given synthetic models
over a range of easy and difficult settings. Overall,the difficulty of the 
task can be controlled by modifying $\Delta$, the ratio between baseline and treatment 
effect ($\sfrac{W_1}{W_2}$), and the number of significant covariates versus total 
number of decoy covariates. Specifically, we construct our experiments along these three axes in
the following manner.
\begin{itemize}
\item \textit{Noise:} The variance of the noise term is evaluated over the range 
                      $\sqrt{\Delta} \in [1, 8]$.
\item \textit{W1:} The ratio between the scale of the trend term (\emph{baseline coefficient}) 
                     and the interaction term. Here, we fix $W_2=1$ while varying $W_1$ over
                     $[1,5]$. 
\item \textit{Decoy:} Given the fixed number of significant covariates, 1 for \textbf{L} and 
                      \textbf{PC} and 7 for \textbf{NL}, $D$ decoy covariates drawn i.i.d. randomly are 
                      added to the dataset. We vary $D$ over $[1,100]$.
\end{itemize}
      
\paragraph{Performance metrics.}
Each of the evaluated significance tests are compared based on their 
statistical power when their Type-I errors
are fixed. Specifically, after generating synthetic data several times and performing
all with a $p$-value 0.05 significance threshold,  we compare the
significance tests by their ability to detect the known significant covariates in 
each synthetic model.
 
Since the test power is a function of many parameters describing various aspects
of signal-to-noise ratio, in our experiments we also trace test sensitivity as a
function of the experimental axes described in the previous section. We then
compare the resulting curves to estimate which test provides the best
performance over the tested range. 
Fig.~\ref{fig:indiv-param} shows an example of such curves for
different significance tests performed on the \textbf{PC-Int}$_1$ model under
varying noise and baseline strength.
As an aggregate measure for overall test performance,
 we compute the area under each curve. In this particular example,
\textit{MaxBE} is the best performing significance test, as it is 
capable of detecting the signal more often than alternative approaches. 
\begin{figure}[t]
    \centering
    \includegraphics[width=0.49\textwidth]{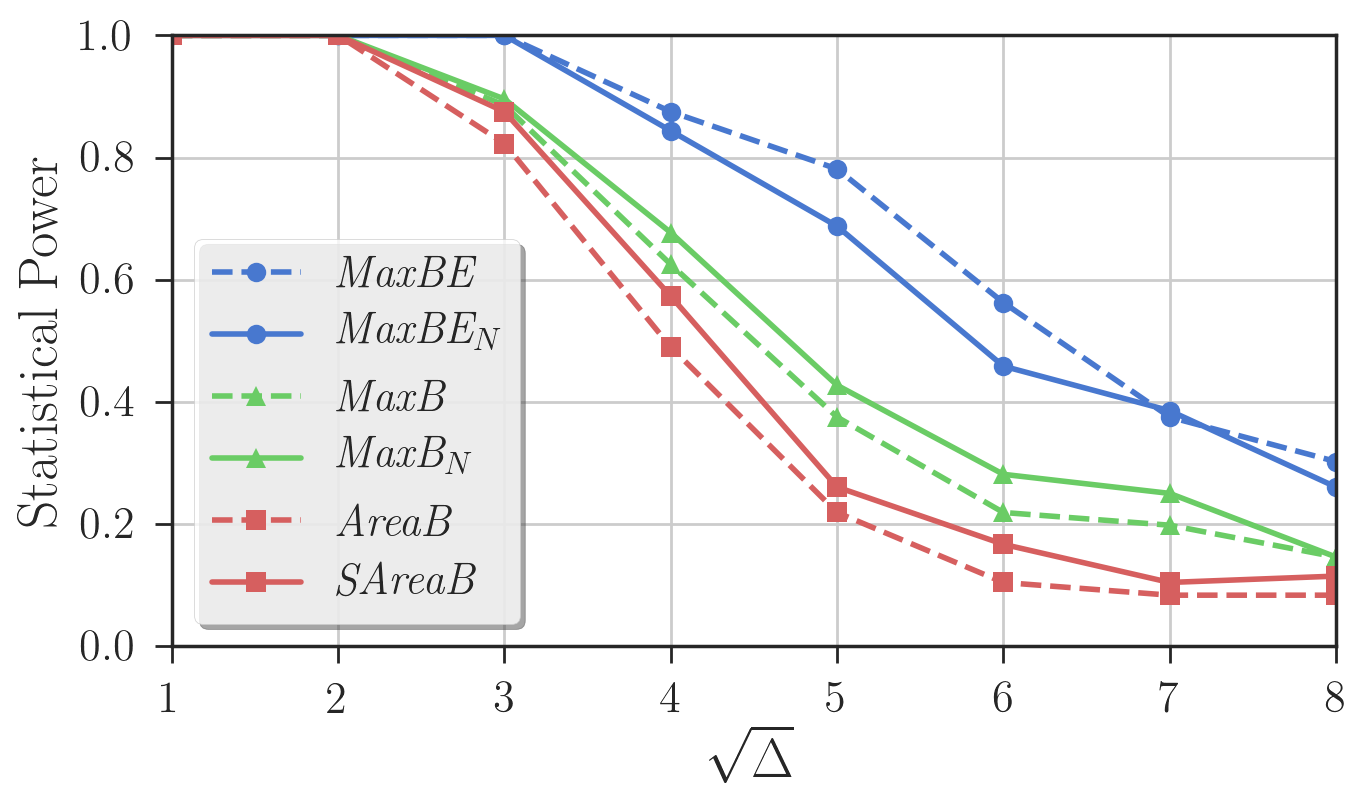}
    \includegraphics[width=0.49\textwidth]{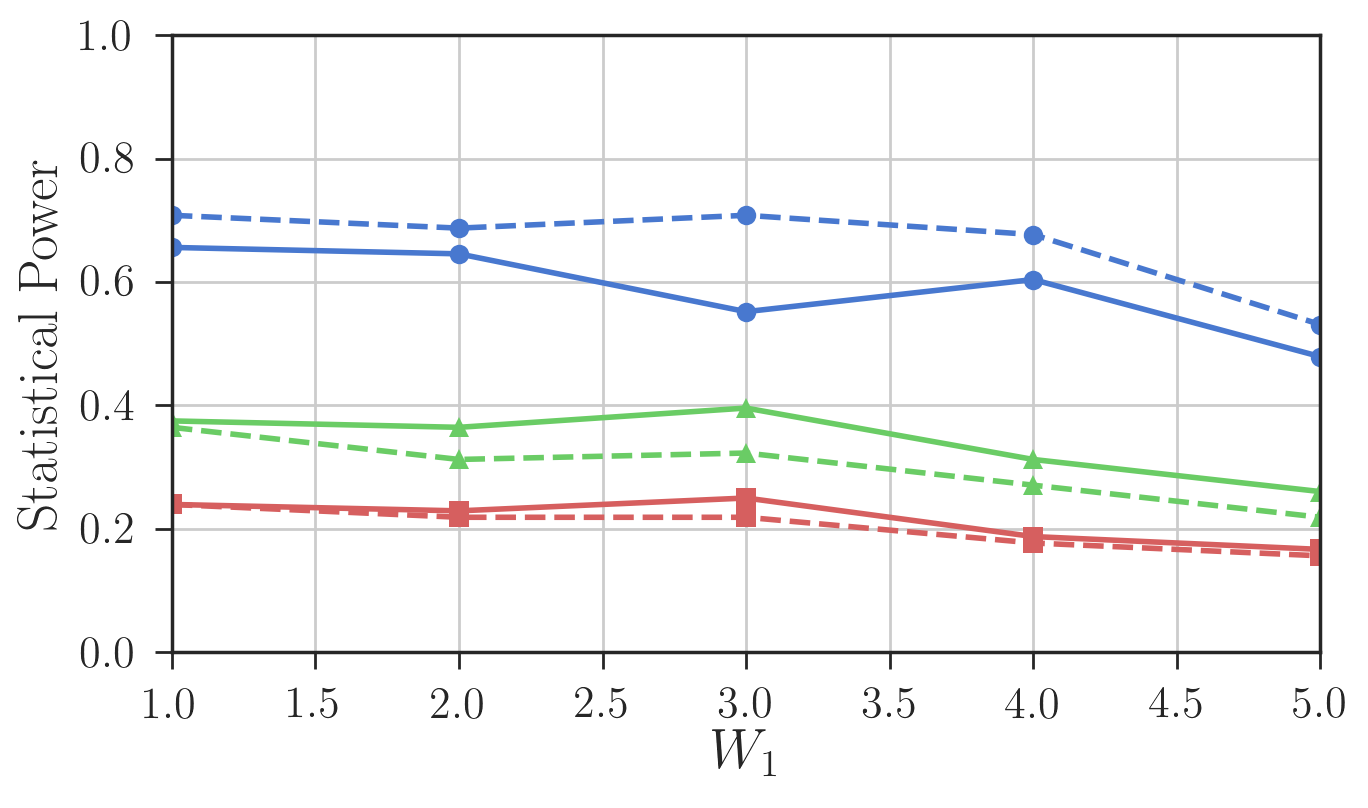} 
    \caption{%
        Comparison of individual tests for model {\bf PC-Int}$_1$ under varying
        experimental parameters. \textit{Left} Comparison over noise
        standard deviation, $\sqrt{\Delta}$. \textit{Right:} Comparison over the value of
        the model baseline coefficient $W_1$. 
        The legend remains consistent over both charts.
        \label{fig:indiv-param}}
\end{figure}

\subsection{Results}    	        	
\paragraph{Centering effect.}
Fig.~\ref{fig:BriMot} shows the improvement made by removing the average
values on each group of patients, comparing only the \emph{Max} tests 
performed on either a Brownian-Motion-like process (no centering) 
or on a Brownian-Bridge-like process (centering). For visual convenience, 
the scores are normalized by the maximum value  per dataset (column) such that 
there is alway a test with the performance score equal to 1 in each simulation scenario. 
As expected from our theoretical analysis, we observe that the test with centering correction 
(\emph{MaxB}), in general, provides better significant covariate detection than the test with
no centering correction (\emph{Max}).

\begin{figure}[t]
    \centering
    \includegraphics[width=0.7\textwidth,trim={0, 3.0cm, 0, 2.7cm},clip]{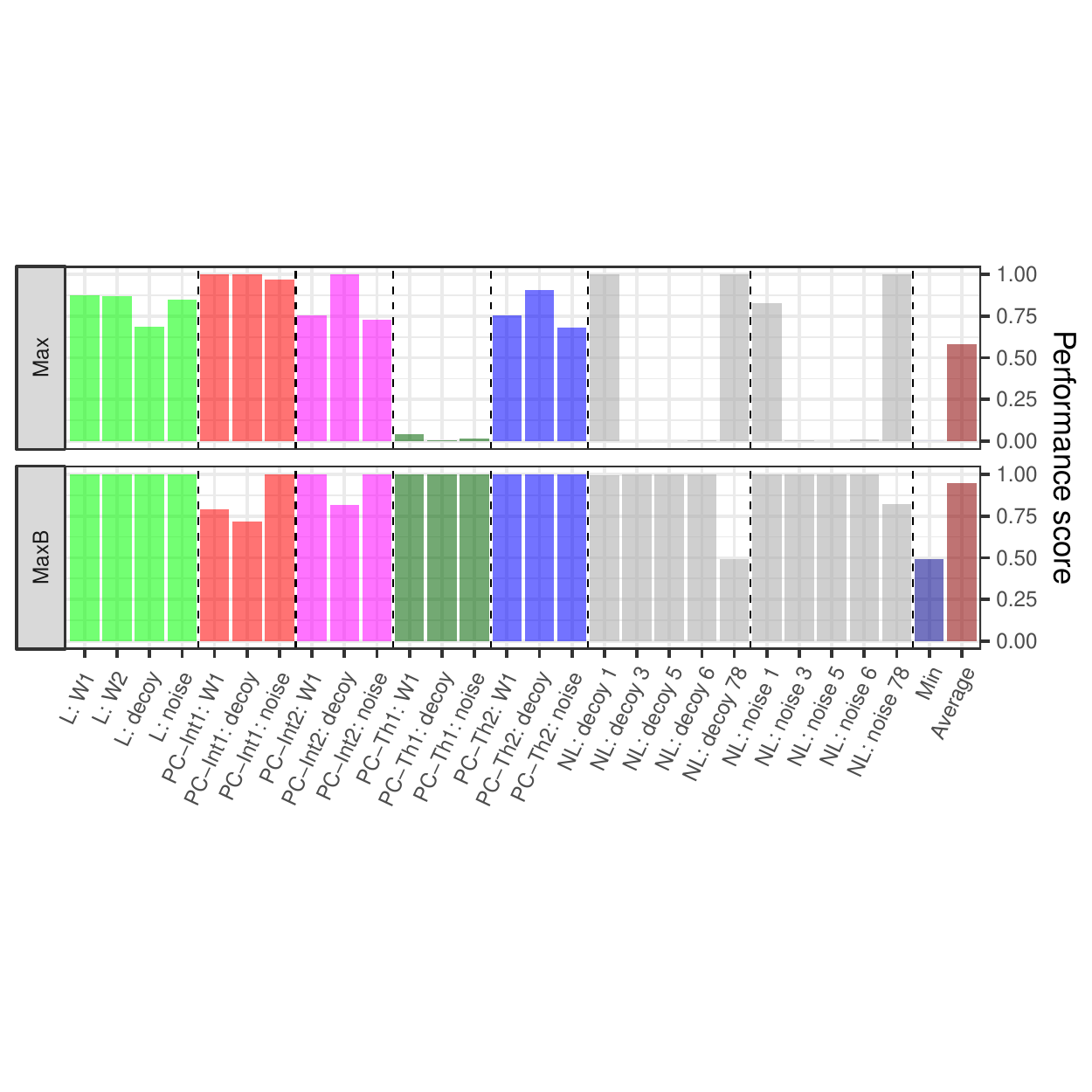}
    \caption{%
        Detailed performance comparison between Brownian Motion (un-centered
        response) \textbf{Max} and Brownian Bridge (centered response)
        \textbf{MaxB} tests over the evaluated models and experimental
        parameters.  }
     \label{fig:BriMot}
\end{figure}

\paragraph{Main result summary.}    	  
A global summary of the main results is presented in Fig.~\ref{fig:hist} for
the cumulative process tests \textit{MaxB}, $\textit{MaxB}_N$, \textit{MaxBE},
$\textit{MaxBE}_N$, \textit{AreaB}, \textit{SAreaB}, the baseline linear
correlation test with modified outcome \textit{MoLin},
and also for the combined test \textit{Comb}.  On this chart, each row
corresponds to the designated test, 
while each column corresponds to a particular synthetic model and the varied parameter
(\emph{noise}, \emph{W1}, or \emph{decoy}), as described in Sec.~\ref{sec:exp-param}.

The first set of columns corresponds to model \textbf{L}.
We report that most of the tests have the same statistical power for the detection of
the significant covariate $X^1$ for this model.
As expected, the linear test \textit{MoLin} performs very well,
similarly to tests which are sensitive to a global measurements on the cumulative
process, such as \emph{AreaB} and \emph{SAreaB}, since averaging over the whole curve provides
more robustness. The \textit{MaxBE} test does not perform well on linear data, as it is
designed to capture transient correlations. Because of this feature, for \textbf{L},
it mostly detects unrelated fluctuations. 
The combined test \emph{Comb} is also shown to perform very well in the linear setting.

The next four sets of columns correspond to the piecewise constant models (\textbf{PC}) and 
their associated experimental axes. 
For this class of models, the tests under study perform very differently, 
revealing the strengths of each individual test, which we now describe.

\textit{i) Threshold detection.}
As expected, \emph{MaxB} performs the best for detection of threshold effects, at 
least when threshold is centrally located (\textbf{PC-Th}$_{1}$). Actually, this task appears relatively 
easy since most tests report good power, even \emph{MoLin}. However, when the
threshold is closer to the boundaries, and therefore more difficult to detect as in \textbf{PC-Th}$_{2}$, 
the \emph{MaxB} test fails to recognize the significant covariate, as do most of the other tests.
However, \emph{MaxBN} performs very well in this setting due to the
normalization hull which brings more power to the extreme sides of. 

\textit{ii) Interval detection.}
This task appears to be very difficult for most tests, which fail
catastrophically on both \textbf{PC-Int}$_{1}$ and \textbf{PC-Int}$_{2}$. 
The \emph{MaxB} and \emph{$MaxB_N$} tests can only detect a
single jump discontinuity, while \emph{AreaB}, \emph{SAreaB}, and \emph{MoLin} 
are global measures which are unable to capture such a localized phenomenon.  
The only two tests which give strong positive results on this task is the
\emph{MaxBE} test and its normalized version \emph{$MaxBE_N$}, 
since large excursions of the cumulative process correspond to a transient effect on the 
covariates. 

We also see that \emph{Comb} demonstrates very robust behavior across all the tested models
and experimental parameters.
While it does not always provide the best possible performance,
its strength lies in its ability to operate in widely varied settings, from
global linear trends to narrow threshold or interval effects.

Finally, to see how those tests would perform on less artificial models, the next set of columns 
correspond to experiments performed on the \textbf{NL} model. Here, we report the detection 
performance for each of the significant variables in \textbf{NL}.
Once again, we observe that particular experimental configurations
of this model can lead to individual test failures. However, once again we see that
\emph{Comb} demonstrates very robust performance across all tested experimental parameters.

These synthetic experiments have illustrated the ability of each individual test to detect
known significant covariates under 
specific assumptions. When those assumptions are known \emph{a priori} by the investigator (e.g. a known
threshold effect), then it may be wiser to select the individual test according to this knowledge. 
However, when this knowledge is not available or too uncertain, our analysis suggests that 
\emph{Comb} could be a robust solution. We present our full set of results both in terms of normalized
AUC as well as raw statistical power in Appendix \ref{apdx:full-results}.

\begin{figure}[H]
    \centering
    \includegraphics[width=0.7\textwidth]{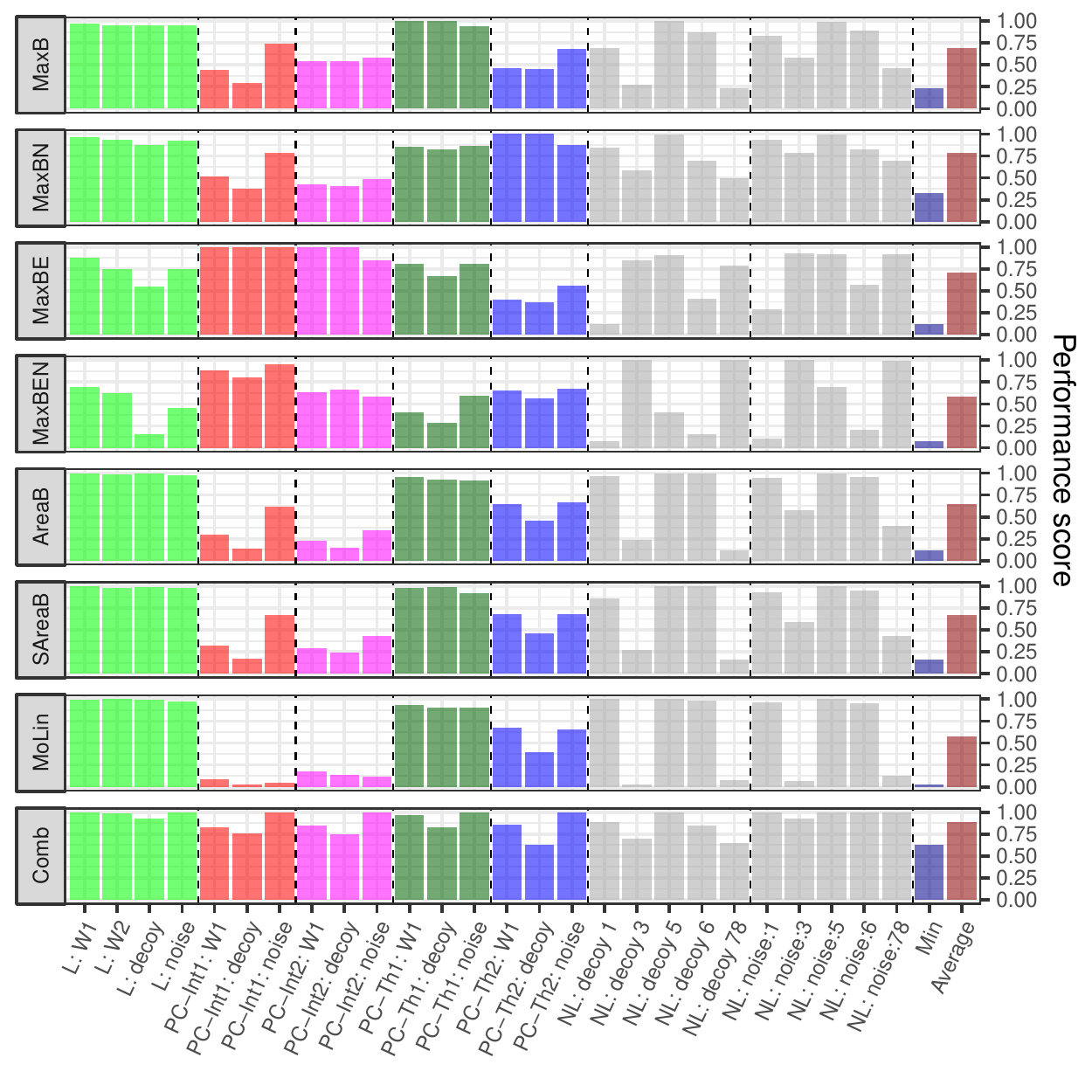} 
    \caption{%
        Visual representation of results obtained from all proposed statistical tests, along with
        the baseline \emph{MoLin}, when performed over the synthetic treatment-interaction models
        (\textbf{L}, \textbf{PC}, and \textbf{NL})
        described in Sec.~\ref{sec:synth-exp}. 
        Evaluated significance tests are given as rows while synthetic models and the experimentation
        parameter are given as columns. Colors are used to indicate different synthetic models.
        Presented performance scores are given as the
        normalized area under the curve over the tested experimental parameter 
        (\emph{noise}, $W_1$, \emph{decoy}). Normalization is performed over the significance tests
        such that the best-performing test reports a performance score of 1.
        The final two columns represent the aggregate minimum and average
        performance of each significance test over the set of tested models and experimentation 
        parameters. Each significance is shown to perform well on some models but worse on others. 
        Notably, the combined test (\emph{Comb}, last row), shows robust performance across all
        experiments.
        \label{fig:hist}}
\end{figure}

\section{Exploration on Real Trials}
\label{sec:realexp}
We now turn our attention to the utility of our new individual tests and their combination
in a more realistic setting.  In this section, we
describe results of application of our method to several real world clinical
trial datasets. For these datasets, there are no known ground truths for significant
covariates. However, as we control Type-I error, we can deduce that detecting more
significant covariates is a desirable property for a successful statistical test on 
these real datasets. 

\subsection{CALGB 40603 (NCT00861705)}
The primary objective of this study was to investigate whether adding
bevacizumab to paclitaxel (+/- carboplatin) and subsequent dose-dense
doxorubicin and cyclophosphamide (ddAC) significantly raises the rate of
pathologic complete response (pCR) in the breast of patients with
HR-poor/ HER2(-), resectable breast cancer \cite{Abramson:2014aa}. The patient
level data of this study can be retrieved from 
Project DataSphere%
\footnote{\url{https://www.projectdatasphere.org/projectdatasphere/html/content/162}} \cite{Green:2015aa}.
The dataset contains the information on
443 patients randomly assigned to four different arms,
\begin{enumerate}
\item paclitaxel $\rightarrow$ ddAC, 
\item paclitaxel + bevacizumab $\rightarrow$ ddAC + bevacizumab, 
\item paclitaxel +  carboplatin $\rightarrow$ ddAC,
\item paclitaxel + carboplatin + bevacizumab $\rightarrow$ ddAC + bevacizumab. 
\end{enumerate}
The dataset contains 45 covariates that can be used to explain the treatment effect
when comparing the 4 arms. 

Table \ref{tab_res_compiled} provides the number of significant variables (0.05
$p$-value significance threshold after Bonferoni correction) detected by the various
tests when comparing Arm 1 versus 2, 1 versus 3, etc. For example, there two
significant interactions when we compare Arms 2 and 4,  the \textbf{clinical N-stage}
(characterization of  the regional lymph node involvement) and \textbf{clinical T-stage}
(characterization of  the size and extent of the tumor). We see that these
covariates are detected only by the combined test,  meaning that the application of
individual test would not have been sufficient for the detection of this
interaction. A similar phenomenon is observed when we analyze the results of the
comparison of Arms 3 and 4. Two significant descriptors are detected: 
\textbf{clinical N-stage} and \textbf{number of sentinel nodes examined}.  
Again, only the combined test
was capable of detecting both interactions simultaneously. Fig.~\ref{fig_calgb}
shows examples of the cumulative curves corresponding to significant variables for
the comparison of Arms 3 and 4.

\begin{figure}[H]
    \centering
    \includegraphics[width=0.49\textwidth]{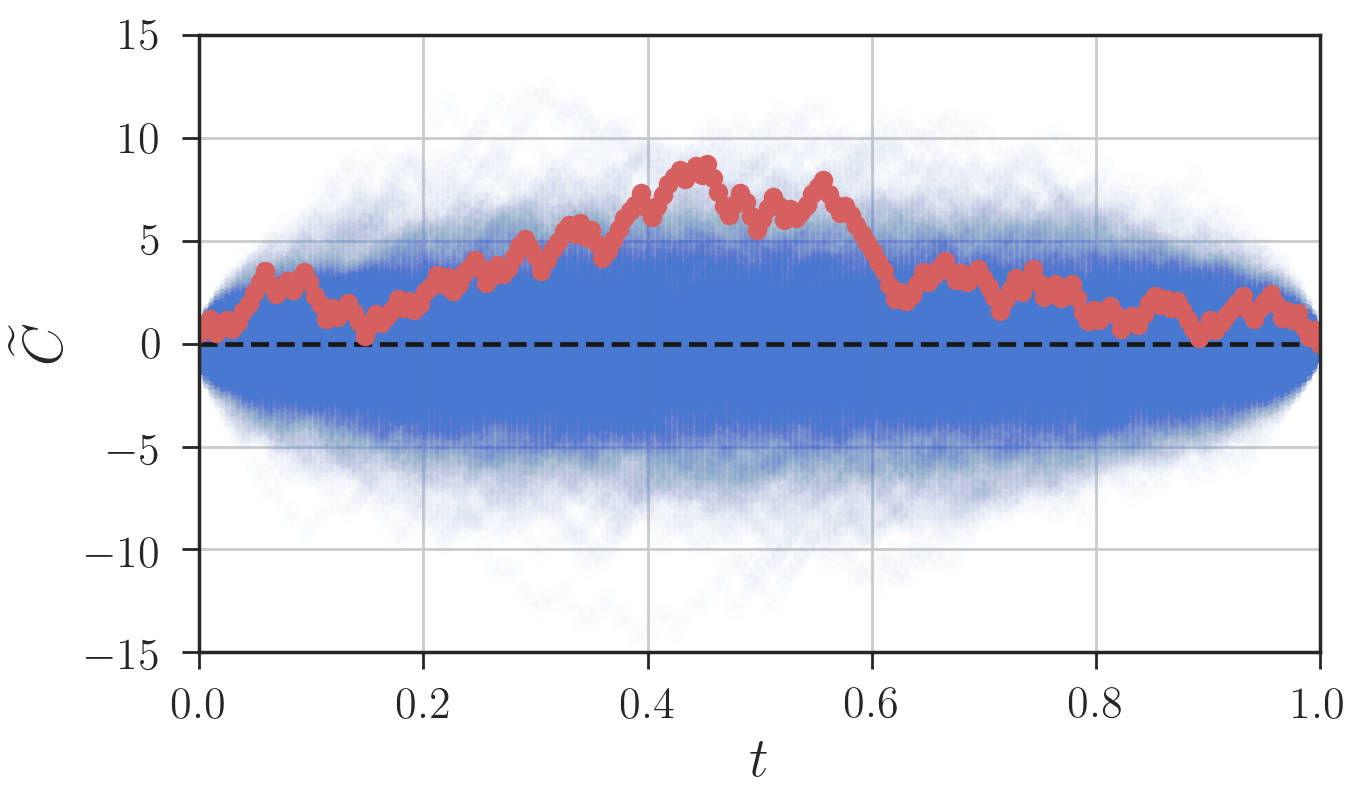}
    \includegraphics[width=0.49\textwidth]{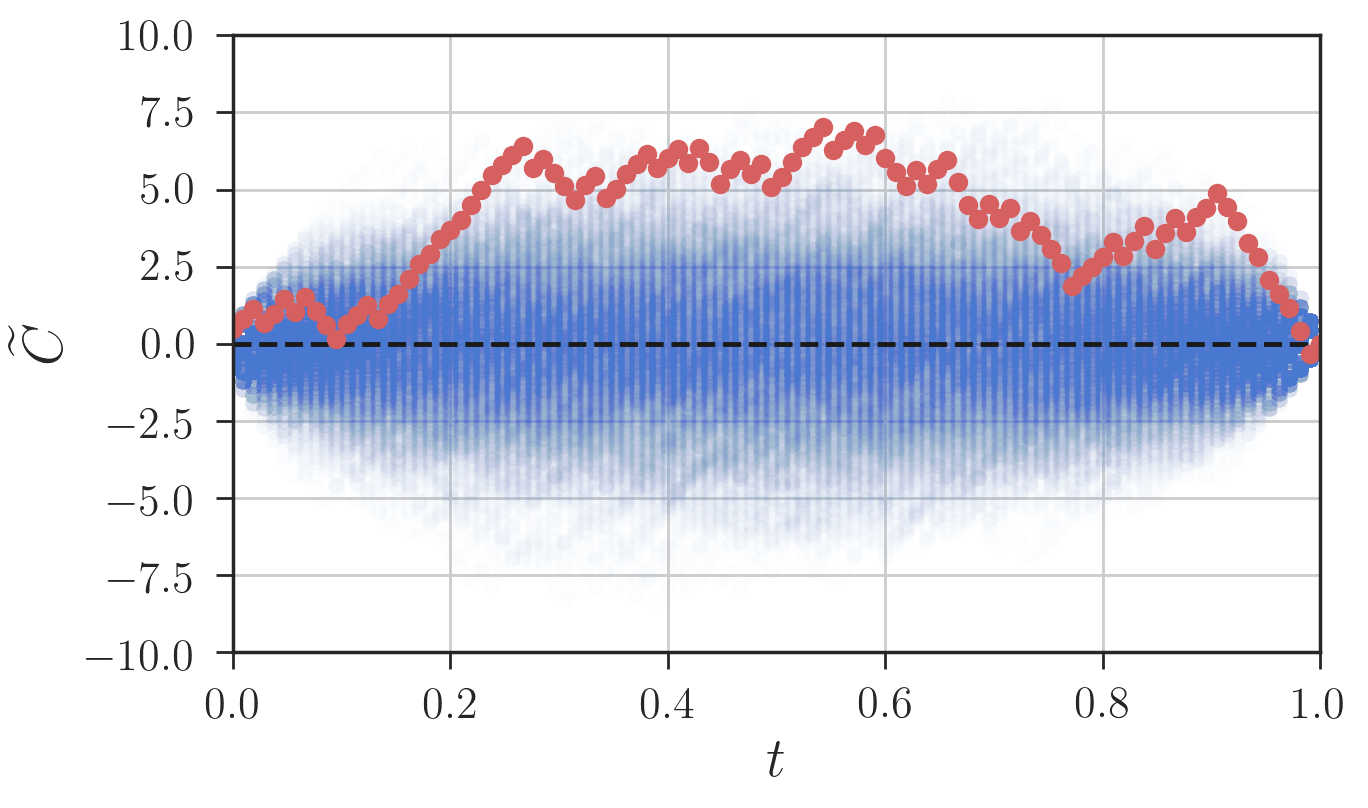}
      \caption{CALGB 40603 dataset,  Arm 4 versus Arm 3. 
      Cumulative processes generated 
      from true covariate data (\emph{red}) is contrasted with those
      constructed from random patient permutations (\emph{blue}).
      \emph{Left:} Evaluation of the \textbf{clinical N-stage} covariate.
      \emph{Right:} Evaluation of the \textbf{number of sentinal nodes examined} covariate.
      \label{fig_calgb}}
    \end{figure}

In both cases, the impact is negative, meaning that smaller covariate values
correspond to larger differences between treatment and placebo (or reference
treatment) arms, therefore implying that patients with smaller values of
these covariates are more likely to benefit from the administration of 
bevacizumab. Since each of the detected covariates characterize the complexity of
the tumor, the discovered dependencies suggest that patients at early stages of the
disease are more likely to benefit from the addition of bevacizumab.

\subsection{BCRP}
The objective of this study was to compare the impact of nutritional and
educational interventions on the psychological and physical adjustment after
treatment for early-stage breast cancer \cite{Scheier:2007aa}.  The dataset was
retrieved from the Quint package \cite{Dusseldorp:2016aa}. In total, there are
252 patients split into three arms: Arm 1: nutrition intervention ($N=85$), Arm 2:
educational intervention ($N=83$), Arm 3: standard care ($N=84$). Similar to the
CALGB results, as shown in Table~\ref{tab_res_compiled}, 
the combined test was the most sensitive test (0.05 $p$-value 
threshold after Bonferoni correction), indicating that the
\textbf{nationality} covariate is a significant factor in determining 
the physical functioning score (SF36). 

\subsection{ACOSOG Z6051 (NCT00726622)}
This is a randomized Phase-III trial evaluating the safety and efficacy of
laparoscopic resection for rectal cancer \cite{Fleshman:2015aa}.  The 462
patients were split into two groups: Arm 1 (open rectal resection) and Arm 2
(laparoscopic rectal resection). Among all tests, as shown in
Table~\ref{tab_res_compiled}, the combined test and extermal of the 
normalized Brownian Bridge test were capable of detecting one covariate significantly correlated
with the efficacy of the treatment regime (0.05 $p$-value significance threshold
after Bonferoni correction): \textbf{distance to nearest radial margin}. The
cumulative process corresponding to this variable is shown in
Fig.~\ref{fig_acosog}, where patients with lower values of the distance to nearest
radial margin covariate are shown to be more likely to benefit more from the
laparoscopic rectal resection. The \textbf{distance to nearest radial margin}
variable is not a parameter that can be assessed prior to the intervention, it is
a characteristic  that can be measured once the operation is complete. Therefore, it
can not be used to guide the selection of the intervention type. However, such 
post-treatment dependencies can provide useful insights into the circumstances
where a treatment of interest may prove most efficient.

\begin{figure}[H]
    \centering
    \includegraphics[width=0.49\textwidth]{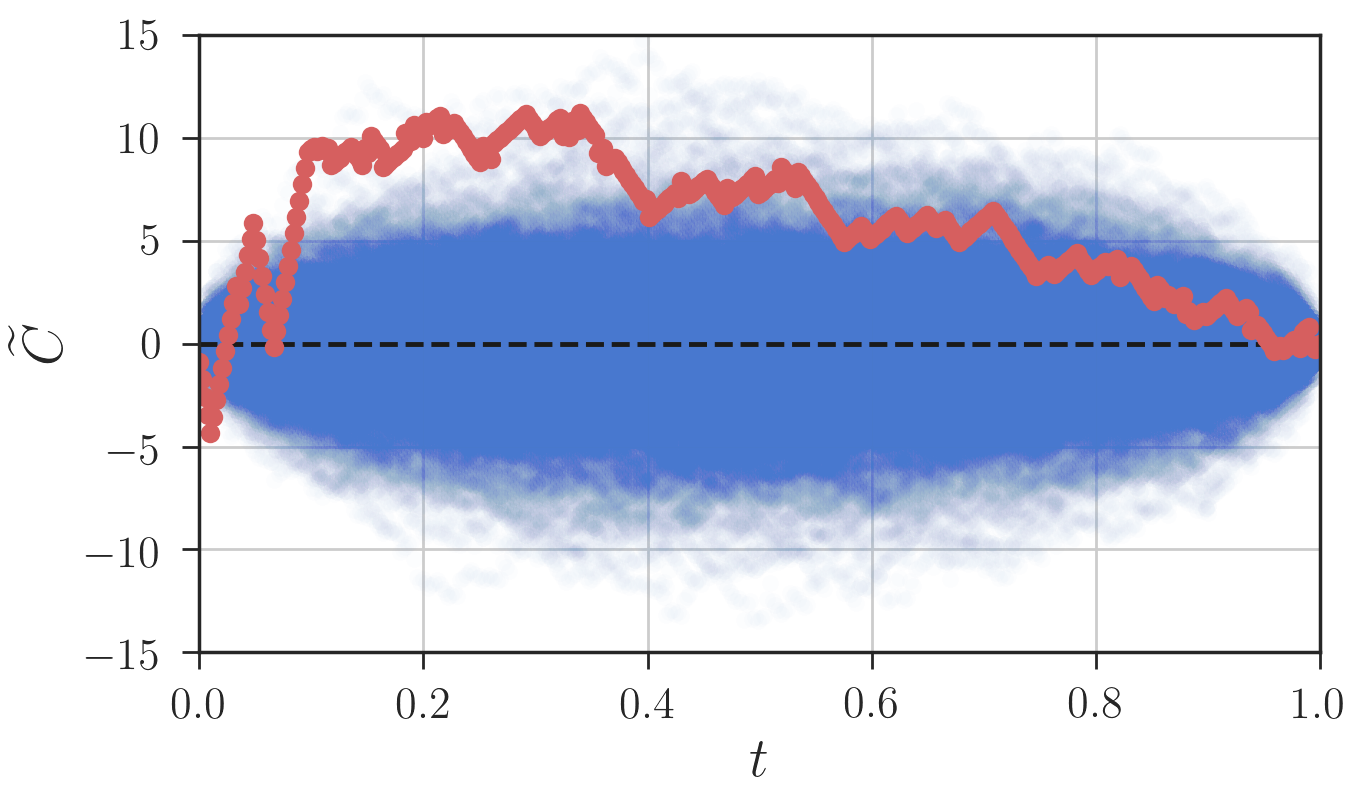}
  \caption{ACOSOG Z6051 dataset evaluated for the 
  \textbf{distance to nearest radial margin} covariate.
  Cumulative processes generated 
  from true covariate data (\emph{red}) is contrasted with those
  constructed from random patient permutations (\emph{blue}).
  \label{fig_acosog}}
\end{figure}

\begin{table}[ht]
    \centering
    \begin{adjustbox}{max width=\textwidth}
    \begin{tabular}{r*{16}{c}}
    \toprule
  & \multicolumn{5}{c}{\bf CALGB}& &\multicolumn{2}{c}{\bf BCRP (C)}& &\multicolumn{2}{c}{\bf BCRP (P)}& &\multicolumn{1}{c}{\bf ACOSOG}&\\
\cmidrule{2-6}                
\cmidrule{8-9}
\cmidrule{11-12}
\cmidrule{14-14}
{\bf Test} &{1v2}&{1v3}&{1v4}&{2v4}&{3v4}& &{1v3}&{2v3}& &{1v3}&{2v3}& &{1v2}&{\bf Total} \\
 \midrule
  \textit{MoLin}    & 0 & 0 & 0 & 0 & 0 & & 0 & 0 & & 0 & 0 & & 0 & 0 \\ 
  \textit{Max}      & 0 & 0 & 0 & 0 & 0 & & 0 & 0 & & 0 & 0 & & 0 & 1 \\ 
  $\textit{MaxB}_N$ & {\bf 1} & 0 & 1 & 0 & 1 & & 0 & 0 & & 0 & 0 & & {\bf 1} & 4 \\ 
  \textit{MaxBE}    & 0 & 0 & 0 & 0 & 0 & & 0 & 0 & & 0 & 0 & & 0 & 0 \\ 
  \textit{AreaB}    & 0 & 0 & 0 & 0 & 1 & & 0 & 0 & & 0 & 0 & & 0 & 1 \\ 
  \textit{SAreaB}   & 0 & 0 & 0 & 0 & 1 & & 0 & 0 & & 0 & 0 & & 0 & 1 \\ 
  \textit{Comb}     & {\bf 1} & 0 & {\bf 2} & {\bf 2} & {\bf 2} & & 0 & 0 & & 0 & {\bf 1} & & {\bf 1} & {\bf 9} \\ 
   \bottomrule
\end{tabular}
\end{adjustbox}
\caption{Number of significant features detected by various statistical tests in
  CALGB 40603, BCRP and ACOSOG Z6051 studies.  BCRP (C) and (P) correspond to
  two different endpoints defined in the study: CESD score and SF36 scores.
  Each row corresponds to a particular test. See Table
  \ref{tab:test-descriptions} for test name abbreviations. 
  Columns names  ``AvB''
  describe comparisons between different arms of the study; 
  see text for a more detail descriptions of study arms. The last column ({\bf Total})
  reports the total number of significant variables detected by each test across all
  evaluated studies.} 
\label{tab_res_compiled}
\end{table}

\section{Discussion}
\label{sec:discussion}
In this paper, we presented a series of new statistical tests tailored for the
detection of complex, non-linear treatment-covariate dependencies in clinical
trial datasets.  We describe how to merge these tests into a single combined test
capable of efficiently detecting different types of interactions. We illustrated the
performance of the proposed approach on various synthetic models, where we
compare it to existing approaches. We also describe an application of the
proposed procedure to three real world clinical trial datasets where we observed
that our proposed tests do indeed allow one to detect signals which can go undiscovered 
using existing approaches.

The proposed technique is a univariate procedure tailored for the detection of
single biomarkers, and therefore is relatively easy to interpret with respect to
multivariate approaches. However, this characteristics is also a significant
limitation of the procedure since, as with any univariate approach, one needs to
apply multiple testing correction (e.g. Bonferoni) to account for the
number of  covariates analyzed. In some cases, this correction might be too
drastic. This issue can be partially resolved by imposing a pre-order on the list
of covariates defined by an expert or by automatically extracting an ordering
from the literature or an external dataset. Another consequence of the nature of 
univariate tests is that 
such a test can miss complex dependencies involving multiple covariates if each of them
separately do not present statistically significant signals.

In our future research, we plan generalize our proposed
cumulative tests to a multivariate setting, and explore the possibility of building a
statistical test capable of detecting non-linear multivariate dependencies with a
proper control over Type-I error. Another interesting direction is the
incorporation of external data sources as prior knowledge. By continuing to
analyze publicly available clinical trial datasets, we hope to 
find new and interesting covariate dependencies, as well as to construct
a large benchmark of datasets for the comparison of covariate-detection methods.   

In this article, we focused on the detection of covariates interacting with the
treatment efficacy. Such covariates provide a natural basis for sub-group
selection. The extension of the proposed statistical tests for sub-group
identification is another important direction that we hope to address in
our future work.        

\bibliographystyle{ieeetran}
\bibliography{references}

\clearpage
\renewcommand{\thesection}{\Alph{section}}
\setcounter{table}{0}
\setcounter{figure}{0}
\renewcommand{\thetable}{\Alph{section}\arabic{table}}
\renewcommand{\thefigure}{\Alph{section}\arabic{figure}}
\appendix
\section{Outcome Centering Proofs}
\subsection{Proof of Lemma \ref{lemma:var-reduce}}
\label{apdx:proof-var-reduce}
As stated in the lemma, we assume that a ``general'' modified outcome
written in the form $\widehat{Y} = T\cdot(R + f(T))$, where
$f(T)$ is an arbitrary function of $T$. Looking at the variance of 
this general modified outcome we have
\begin{equation}
\mathbf{var}[\widehat{Y}]  = 
  \mathbf{var}\left[T\cdot(R + f(T))\right].
\end{equation}
Using the law of total variation, we next condition on the treatment, $T$,
\begin{align}
\mathbf{var}[\widehat{Y}]  = 
&\mathbb{E}\left[\mathbf{var}_R\left[T\cdot(R + f(T))| T\right] \right] \notag \\
&+\mathbf{var}\left[ \mathbb{E}_R\left[ T\cdot(R + f(T)) | T\right]\right].
\end{align}
The conditional variance can be written as
\begin{align}
\mathbf{var}_R\left[T\cdot(R + f(T))| T\right]  &=
T^2 \cdot\mathbf{var}\left[R | T\right],
\end{align}
which follows from the application of the translational invariance property as well as 
the scaling identity. Subsequently, for the conditional expectation, we find
\begin{align}
\mathbb{E}_R&\left[ T\cdot(R + f(T)) | T \right] \notag \\
&=\mathbb{E}_R\left[T\cdot R | T \right] + \mathbb{E}_R\left[ T\cdot f(T) | T\right] \notag\\
&= T\cdot(\mathbb{E}\left[R | T\right] + f(T)),
\end{align}
which follows from both the linearity of expectation as well as the independence of 
$T$ from $R$. Taken together, we have
\begin{align}
\mathbf{var}[\widehat{Y}] &=
  \mathbb{E}\left[T^2\cdot \mathbf{var}[R|T] \right] \notag\\
  &\quad+ \mathbf{var}[T\cdot(\mathbb{E}[R|T] + f(T))].
\end{align}
From the above, we see that the only a variance term is dependent upon the choice of 
$f(T)$. Since the variance must not be negative, the minimum variance 
$\sigma^* = \mathbb{E}[T^2\cdot\mathbf{var}[R|T]]$ is obtained when 
$\mathbf{var}\left[T\cdot(\mathbb{E}[R|T] + f(T) )\right] = 0$. It can be 
directly observed that $f(T) = -\mathbb{E}[R|T]$ is one such choice of $f(T)$ for
which $\mathbf{var}[\widehat{Y}] = \sigma^*$. \qed

\subsection{Proof of Lemma \ref{lemma:two-trial-compare}}
\label{apdx:proof-lemma-two-trial-compare}
We first calculate the variance of the modified outcome as presented in \cite{tian2014simple}. Recalling
$Y_{\rm mod} = T\cdot R$, we see
\begin{align}
\mathbf{var}&[Y_{\rm mod}] = \mathbf{var}[R\cdot T], \notag \\
&= \mathbb{E}[\mathbf{var}_R[T\cdot R |T]] + \mathbf{var}[\mathbb{E}_R[T\cdot R | T]], \notag \\
&= \mathbb{E}[T^2 \cdot\mathbf{var}[R|T]] + \mathbf{var}[T\cdot\mathbb{E}[R|T]].
\end{align}
which follows first from the law of total variation, and also from our earlier calculations in 
Appendix \ref{apdx:proof-var-reduce}. 

We now make the calculation of this variance explicit through the empirical per-trial 
first and second moments, $\mu_T$ and $\sigma_T^2$, respectively. For the first term,
\begin{align}
  \mathbb{E}[T^2\cdot\mathbf{var}&[R|T]] \notag\\
  &=\sum_{t\in\{\pm 1 \}} \pi_t \cdot t^2 \cdot \mathbf{var}[R|T=t], \notag \\
  &= \frac{1}{2}\cdot(\sigma_1^2 + \sigma_{-1}^2)^2.
\end{align}
Subsequently,
\begin{align}
\mathbf{var}&[T\cdot\mathbb{E}[R|T]] \notag \\ 
&= \mathbb{E}[(T\cdot\mathbb{E}[R|T])^2] - \mathbb{E}^2[T\cdot\mathbb{E}[R|T]], \notag\\
&= \sum_{t\in\{\pm 1 \}} \pi_t \cdot t^2 \cdot \mathbb{E}^2[R|T=t] \notag\\
&\quad\quad-\left(\sum_{t\in\{\pm 1 \}} \pi_t \cdot t \cdot \mathbb{E}[R|T=t]\right)^2, \notag\\
&= \frac{1}{4}(\mu_1 - \mu_{-1})^2.
\end{align}
Finally, we have
\begin{equation}
  \mathbf{var}[Y_{\rm mod}] = \frac{1}{2}(\sigma_1^2 + \sigma_{-1}^2) + \frac{1}{4}(\mu_1 - \mu_{-1})^2.
\end{equation}
In contrast, we have the value of the modified outcome with per-treatment centering, as shown in
Lemma \ref{lemma:var-reduce},
\begin{align}
  \mathbf{var}[\widetilde{Y}] &= \mathbb{E}[T^2 \cdot\mathbf{var}[R|T]], \notag \\
  &= \frac{1}{2}(\sigma_1^2 + \sigma_{-1}^2).
\end{align}
Thus, for the ratio $\gamma = \frac{\mathbf{var}[\widetilde{Y}]}{\mathbf{var}[Y_{\rm mod}]}$,
\begin{align}
  \gamma &= 
  \frac{\frac{1}{2}(\sigma_1^2 + \sigma_{-1}^2)}{\frac{1}{2}(\sigma_1^2 + \sigma_{-1}^2) + \frac{1}{4}(\mu_1 - \mu_{-1})^2},
  \notag \\
  &= \frac{1}{1 + \frac{1}{2}\cdot\frac{(\mu_1 - \mu_{-1})^2}{\sigma_1^2 + \sigma^2_{-1}}}, \notag \\
  &= 1 - {\rm sig}\left[\log \frac{(\mu_1 - \mu_{-1})^2}{\sigma_1^2 + \sigma_{-1}^2} - \log 2 \right].
\end{align}
where ${\rm sig}(x) = \frac{1}{1+e^{-x}}$ is the logistic sigmoid function. Since ${\rm sig}(x)$ is bounded
in range $[0,1]$ over the domain $x \in \mathbb{R}$, we can see that $\gamma \in [0,1]$, as well. Additionally,
we see directly that the only instance in which $\gamma = 1$ is the case that $\mu_1 = \mu_{-1}$, that is,
when the measured responses for both treatments have the same expectation. \qed

\subsection{Proof of Lemma \ref{lemma:centering}}
\label{apdx:proof-lemma-centering}

The conditional covariance of the centered outcome can be shown to be 
\begin{align}
\mathbf{cov}&[ \widetilde{R}, T | X] \notag\\
&=    \mathbf{cov}\left[ R - \mathbb{E}[R|T] , T | X\right], \notag \\
&= 
    \mathbf{cov}[R,T|X] - \mathbf{cov}[\mathbb{E}[R|T], T|X], \notag \\
&=
    \mathbf{cov}[R,T|X] - \mathbb{E}[T\cdot \mathbb{E}[R | T] | X],   \notag \\
    &\quad + \mathbb{E}[\mathbb{E}[R|T]|X]\cdot \mathbb{E}[T|X].
\end{align}
Assuming a properly randomized trial such that $T \perp X$ and centered treatment variables, 
$\mathbb{E}[T] = 0$, then the covariance becomes
\begin{align}
\mathbf{cov}&[ \widetilde{R}, T | X] \notag \\
&= \mathbf{cov}[R,T|X] - \mathbb{E}[T\cdot \mathbb{E}[R|T]|X].
\end{align}
However, since $T\perp X$ and the measured outcome $R$ is a fixed property of 
the dataset which does not change based on the covariate selection $X$, we
note that $\mathbb{E}[T\cdot \mathbb{E}[R|T]|X] = \mathbb{E}[T\cdot \mathbb{E}[R|T]]$.
Thus, the second term is constant with respect to the covariate under test and 
we write it as $\mathcal{C}_{R,T}$. \qed

\subsection{Proof of Lemma \ref{lemma:expected-covariance}}
\label{apdx:proof-lemma-expected-covariance}
The unconditioned covariance between the centered treatment response and the treatment 
indicators is
\begin{align}
\textbf{cov}&[\widetilde{R}, T] \notag\\
    &= \textbf{cov}[R - \mathbb{E}[R|T], T], \notag \\
    &= \textbf{cov}[R,T] - \mathbf{cov}[\mathbb{E}[R|T], T], 
\end{align}
which follows from the simple distributional property of the covariance. Since we assume 
centered treatment indicator variables, $\mathbb{E}[T] = 0$, this difference in covariances can
be simplified to the difference,
\begin{equation}
\textbf{cov}[\widetilde{R}, T] = 
    \mathbb{E}[T\cdot R] - \mathbb{E}[T\cdot\mathbb{E}[R|T]].
\end{equation}
However, we can see that these two terms are equivalent. Using the total expectation, we observe
that $\mathbb{E}[T\cdot R] = \mathbb{E}_T[\mathbb{E}_R[T\cdot R| T]]$, where we use subscripts 
to make the expectations explicit. Moving the constant $T$ outside of the conditional expectation,
we have $\mathbb{E}[T\cdot R] = \mathbb{E}[T\cdot \mathbb{E}[R|T]]$, which shows the 
equivalence of the two terms. Thus, $\textbf{cov}[\widetilde{R}, T] = 0$. \qed

\clearpage
\section{Full Synthetic Results}
\label{apdx:full-results}
\begin{table}[h]
\centering
  \begin{adjustbox}{max width=4.1in}
  \begin{tabular}{l*{8}{c}} \toprule
    {Synthetic Model} & {\it MoLin}  & {\it AreaB} & {\it SAreaB} & {\it MaxB} & {\textit{MaxB}$_N$} & {\it MaxBE} & {\textit{MaxBE}$_N$} & {\it Comb} \\
    \midrule
    \input{tab_res_real_values}
    \bottomrule
    \end{tabular}
\end{adjustbox}
  \caption{Full table of raw statistical power for synthetic tests. 
             \label{tab:power}}  
\end{table}

\begin{table}[h]
\centering
\begin{adjustbox}{max width=4.1in}
  \begin{tabular}{l*{8}{c}} \toprule
    {Synthetic Model} & {\it MoLin}  & {\it AreaB} & {\it SAreaB} & {\it MaxB} & {\textit{MaxB}$_N$} & {\it MaxBE} & {\textit{MaxBE}$_N$} & {\it Comb} \\
    \midrule
    \input{tab_res_area_norm}
    \bottomrule
    \end{tabular}
\end{adjustbox}
  \caption{Full table of normalized area for synthetic tests. 
             \label{tab:area_norm}}  
\end{table}

\end{document}

%% file: tab_res_real_values.tex
{\textbf{L}}  &  \multicolumn{8}{c}{{}} \\ 
\cmidrule(lr){1-1}
{Base ($W_1$)} & 0.958 & 0.979 & \gudr{0.990} & 0.896 & 0.896 & 0.604 & 0.396 & {\textbf{0.990} (0.0\%)}\\
{Treatment ($W_2$)} & \gudr{0.969} & 0.948 & 0.927 & 0.885 & 0.854 & 0.510 & 0.250 & {0.896 (-7.5\%)}\\
{Noise} & 0.865 & 0.865 & \gudr{0.875} & 0.792 & 0.771 & 0.490 & 0.177 & {0.833 (-4.8\%)}\\
{Decoy} & 0.667 & \gudr{0.719} & 0.677 & 0.531 & 0.469 & 0.094 & 0.031 & {0.562 (-21.8\%)}\\
\multicolumn{9}{c}{{}}\\
{\textbf{PC-Int}$_2$}  &  \multicolumn{8}{c}{{}} \\
 \cmidrule(lr){1-1}
{Noise} & 0.042 & 0.042 & 0.052 & 0.188 & 0.104 & \gudr{0.427} & 0.219 & {0.375 (-12.2\%)}\\
{Decoy} & 0.000 & 0.000 & 0.000 & \gudr{0.010} & 0.000 & 0.010 & 0.000 & {0.000 (-100.0\%)}\\
{Base ($W_1$)} & 0.021 & 0.010 & 0.031 & 0.083 & 0.042 & \gudr{0.177} & 0.083 & {0.125 (-29.4\%)}\\
\multicolumn{9}{c}{{}}\\
{\textbf{PC-Th}$_1$}  &  \multicolumn{8}{c}{{}} \\
 \cmidrule(lr){1-1}
{Noise} & 0.531 & 0.562 & 0.573 & 0.635 & \gudr{0.833} & 0.417 & 0.562 & {0.823 (-1.2\%)}\\
{Decoy} & 0.021 & 0.021 & 0.042 & 0.021 & \gudr{0.062} & 0.031 & 0.042 & {0.052 (-16.1\%)}\\
{Base ($W_1$)} & 0.198 & 0.188 & 0.188 & 0.104 & \gudr{0.333} & 0.073 & 0.177 & {0.292 (-12.3\%)}\\
\multicolumn{9}{c}{{}}\\
{\textbf{PC-Int}$_1$}  &  \multicolumn{8}{c}{{}} \\
 \cmidrule(lr){1-1}
{Noise} & 0.688 & 0.719 & 0.719 & \gudr{0.729} & 0.635 & 0.583 & 0.219 & {\textbf{0.750} (2.9\%)}\\
{Decoy} & 0.292 & 0.292 & 0.302 & \gudr{0.354} & 0.292 & 0.177 & 0.010 & {0.240 (-32.2\%)}\\
{Base ($W_1$)} & 0.802 & 0.802 & 0.823 & \gudr{0.844} & 0.677 & 0.667 & 0.312 & {0.698 (-17.3\%)}\\
\multicolumn{9}{c}{{}}\\
{\textbf{PC-Th}$_2$}  &  \multicolumn{8}{c}{{}} \\
 \cmidrule(lr){1-1}
{Noise} & 0.010 & 0.219 & 0.260 & 0.375 & 0.427 & \gudr{0.781} & 0.688 & {0.771 (-1.3\%)}\\
{Decoy} & 0.000 & 0.000 & 0.000 & 0.021 & 0.010 & \gudr{0.229} & 0.156 & {0.094 (-59.0\%)}\\
{Base ($W_1$)} & 0.042 & 0.219 & 0.229 & 0.312 & 0.365 & \gudr{0.688} & 0.646 & {0.500 (-27.3\%)}\\
\multicolumn{9}{c}{{}}\\
{\textbf{NL}}  &  \multicolumn{8}{c}{{}} \\
 \cmidrule(lr){1-1}
{Noise: $X^7$, $X^8$} & 0.026 & 0.042 & 0.068 & 0.083 & 0.120 & \gudr{0.203} & 0.203 & {0.161 (-20.7\%)}\\
{Noise: $X^6$} & \gudr{0.281} & 0.271 & 0.260 & 0.219 & 0.260 & 0.104 & 0.052 & {\textbf{0.375} (33.5\%)}\\
{Noise: $X^5$} & \gudr{1.000} & 1.000 & 1.000 & 0.990 & 0.969 & 0.698 & 0.312 & {\textbf{1.000} (0.0\%)}\\
{Noise: $X^3$} & 0.062 & 0.146 & 0.135 & 0.135 & 0.292 & 0.562 & \gudr{0.615} & {0.479 (-22.1\%)}\\
{Noise: $X^1$} & \gudr{0.323} & 0.302 & 0.302 & 0.302 & 0.292 & 0.115 & 0.083 & {\textbf{0.365} (13.0\%)}\\
{Decoy: $X^7$, $X^8$} & 0.000 & 0.000 & 0.000 & 0.000 & 0.005 & 0.031 & \gudr{0.042} & {0.005 (-88.1\%)}\\
{Decoy: $X^6$} & 0.073 & \gudr{0.083} & 0.083 & 0.031 & 0.010 & 0.010 & 0.000 & {0.073 (-12.0\%)}\\
{Decoy: $X^5$} & 0.979 & 0.979 & \gudr{0.990} & 0.938 & 0.927 & 0.688 & 0.167 & {0.969 (-2.1\%)}\\
{Decoy: $X^3$} & 0.010 & 0.010 & 0.010 & 0.031 & 0.125 & 0.302 & \gudr{0.385} & {0.229 (-40.5\%)}\\
{Decoy: $X^1$} & 0.104 & \gudr{0.115} & 0.115 & 0.052 & 0.052 & 0.000 & 0.000 & {0.083 (-27.8\%)}\\

%% file: tab_res_area_norm.tex
{\textbf{L}}  &  \multicolumn{8}{c}{{}} \\ 
\cmidrule(lr){1-1}
{Base ($W_1$)} & 0.994 & 0.998 & \gudr{1.000} & 0.969 & 0.965 & 0.880 & 0.697 & {0.999 (-0.1\%)}\\
{Treatment ($W_2$)} & \gudr{1.000} & 0.990 & 0.980 & 0.953 & 0.938 & 0.753 & 0.618 & {0.967 (-3.3\%)}\\
{Noise} & 0.978 & 0.975 & \gudr{0.982} & 0.950 & 0.925 & 0.750 & 0.455 & {0.976 (-0.6\%)}\\
{Decoy} & 0.991 & \gudr{1.000} & 0.993 & 0.947 & 0.873 & 0.551 & 0.158 & {0.925 (-7.5\%)}\\
\multicolumn{9}{c}{{}}\\
{\textbf{PC-Int}$_2$}  &  \multicolumn{8}{c}{{}} \\
 \cmidrule(lr){1-1}
{Noise} & 0.120 & 0.347 & 0.433 & 0.578 & 0.485 & \gudr{0.839} & 0.578 & {0.838 (-0.1\%)}\\
{Decoy} & 0.034 & 0.037 & 0.060 & 0.135 & 0.102 & \gudr{0.250} & 0.165 & {0.077 (-69.2\%)}\\
{Base ($W_1$)} & 0.062 & 0.081 & 0.103 & 0.190 & 0.150 & \gudr{0.349} & 0.221 & {0.234 (-33.0\%)}\\
\multicolumn{9}{c}{{}}\\
{\textbf{PC-Th}$_1$}  &  \multicolumn{8}{c}{{}} \\
 \cmidrule(lr){1-1}
{Noise} & 0.632 & 0.652 & 0.657 & 0.665 & \gudr{0.849} & 0.547 & 0.652 & {\textbf{0.883} (4.0\%)}\\
{Decoy} & 0.225 & 0.261 & 0.262 & 0.256 & \gudr{0.570} & 0.209 & 0.321 & {0.286 (-49.8\%)}\\
{Base ($W_1$)} & 0.517 & 0.493 & 0.517 & 0.350 & \gudr{0.762} & 0.303 & 0.500 & {0.599 (-21.4\%)}\\
\multicolumn{9}{c}{{}}\\
{\textbf{PC-Int}$_1$}  &  \multicolumn{8}{c}{{}} \\
 \cmidrule(lr){1-1}
{Noise} & 0.899 & 0.914 & 0.917 & \gudr{0.940} & 0.860 & 0.812 & 0.592 & {0.921 (-2.0\%)}\\
{Decoy} & 0.861 & 0.883 & 0.940 & \gudr{0.954} & 0.790 & 0.636 & 0.271 & {0.865 (-9.3\%)}\\
{Base ($W_1$)} & 0.936 & 0.956 & 0.983 & \gudr{1.000} & 0.857 & 0.807 & 0.405 & {0.919 (-8.1\%)}\\
\multicolumn{9}{c}{{}}\\
{\textbf{PC-Th}$_2$}  &  \multicolumn{8}{c}{{}} \\
 \cmidrule(lr){1-1}
{Noise} & 0.043 & 0.588 & 0.638 & 0.699 & 0.740 & \gudr{0.946} & 0.902 & {\textbf{0.959} (1.4\%)}\\
{Decoy} & 0.031 & 0.135 & 0.173 & 0.286 & 0.377 & \gudr{0.990} & 0.797 & {0.705 (-28.8\%)}\\
{Base ($W_1$)} & 0.089 & 0.302 & 0.323 & 0.445 & 0.516 & \gudr{1.000} & 0.880 & {0.710 (-29.0\%)}\\
\multicolumn{9}{c}{{}}\\
{\textbf{NL}}  &  \multicolumn{8}{c}{{}} \\
 \cmidrule(lr){1-1}
{Noise: $X^7$, $X^8$} & 0.130 & 0.400 & 0.431 & 0.461 & 0.692 & 0.920 & \gudr{0.990} & {0.890 (-10.1\%)}\\
{Noise: $X^6$} & 0.954 & \gudr{0.956} & 0.951 & 0.889 & 0.822 & 0.569 & 0.207 & {0.900 (-5.9\%)}\\
{Noise: $X^5$} & 0.999 & \gudr{1.000} & 1.000 & 0.996 & 0.993 & 0.914 & 0.689 & {0.996 (-0.4\%)}\\
{Noise: $X^3$} & 0.062 & 0.580 & 0.591 & 0.578 & 0.783 & 0.931 & \gudr{1.000} & {0.917 (-8.3\%)}\\
{Noise: $X^1$} & \gudr{0.935} & 0.921 & 0.907 & 0.813 & 0.906 & 0.287 & 0.099 & {0.914 (-2.2\%)}\\
{Decoy: $X^7$, $X^8$} & 0.033 & 0.052 & 0.074 & 0.104 & 0.221 & 0.352 & \gudr{0.448} & {0.276 (-38.4\%)}\\
{Decoy: $X^6$} & 0.720 & 0.729 & \gudr{0.730} & 0.639 & 0.510 & 0.297 & 0.113 & {0.647 (-11.4\%)}\\
{Decoy: $X^5$} & \gudr{1.000} & 1.000 & 1.000 & 0.998 & 0.996 & 0.911 & 0.400 & {0.997 (-0.3\%)}\\
{Decoy: $X^3$} & 0.029 & 0.235 & 0.264 & 0.270 & 0.576 & 0.833 & \gudr{0.987} & {0.726 (-26.4\%)}\\
{Decoy: $X^1$} & \gudr{0.720} & 0.696 & 0.616 & 0.500 & 0.611 & 0.088 & 0.053 & {0.630 (-12.5\%)}\\